\documentclass[
    aps,
    prx,
    superscriptaddress,
    twocolumn,
    10pt
]{revtex4-2}

\usepackage{amsmath}
\usepackage{amsthm}
\usepackage{amssymb}
\usepackage{mathtools}
\usepackage{thmtools}
\usepackage{physics}
\usepackage{adjustbox}
\usepackage{enumitem}
\usepackage{tikz}
\usetikzlibrary{quantikz2}
\usepackage{pifont}

\newcommand{\xmark}{\ding{55}}
\newtheorem{problem}{Problem}
\newtheorem{theorem}{Theorem}
\newtheorem{lemma}{Lemma}
\newtheorem{corollary}{Corollary}
\newtheorem{definition}{Definition}
\newtheorem{proposition}{Proposition}
\newcommand{\QMA}{\mathsf{QMA}}
\newcommand{\QCMA}{\mathsf{QCMA}}
\newcommand{\poly}{\mathrm{poly}}
\DeclareRobustCommand{\stirling}{\genfrac\{\}{0pt}{}}
\usepackage{hyperref}
\hypersetup{colorlinks=true, linkcolor=blue, citecolor=blue, urlcolor=blue}
\usepackage[nameinlink,capitalize]{cleveref}

\begin{document}

\title{Provable and Verifiable Quantum Advantage in Sample Complexity}

\date{January 30, 2026}

\author{Marcello Benedetti}
\email{marcello.benedetti@quantinuum.com}
\affiliation{Quantinuum, London, United Kingdom}

\author{Harry Buhrman}
\email{harry.buhrman@quantinuum.com}
\affiliation{Quantinuum, London, United Kingdom}
\affiliation{QuSoft, Amsterdam, The Netherlands}
\affiliation{University of Amsterdam, Amsterdam, The Netherlands}

\author{Jordi Weggemans}
\email{jrw@cwi.nl}
\affiliation{QuSoft, Amsterdam, The Netherlands}
\affiliation{CWI, Amsterdam, The Netherlands}

\begin{abstract}
Consider a fixed universe of $N=2^n$ elements and the uniform distribution over elements of some subset of size $K$. Given samples from this distribution, the task of complement sampling is to provide a sample from the complementary subset. We give a simple quantum algorithm that uses only a single quantum sample---a single copy of the uniform superposition over elements of the subset. When $K=N/2$, we show that the quantum algorithm succeeds with probability $1$, whereas any classical algorithm that succeeds with bounded probability of error requires a number of samples of the order of $N$. This shows that in a sample-to-sample setting, quantum computation can achieve the largest possible separation over classical computation. We show that the same bound can be lifted to prove average-case hardness, paving the way for demonstrations on noisy intermediate-scale quantum (NISQ) computers. It follows that under the assumption of the existence of one-way functions, complement sampling gives provable, verifiable and NISQable quantum advantage in a sample complexity setting.
\end{abstract}

\maketitle

\paragraph*{Introduction ---}
In 2019, Google claimed to have demonstrated quantum advantage for the first time, that is, their Sycamore processor performed a computational task in the order of minutes that Google estimated would take the world’s most powerful classical computer approximately 10,000 years to complete~\cite{arute2019quantum}. Their experiment was based on random circuit sampling where, given a classical description of a quantum circuit, the task is to sample in the computational basis according to the probability distribution associated with the output state. Google's achievement started a flurry of research on both ends of the spectrum: new advantage claims were made based on other sampling experiments~\cite{zhong2020quantum,madsen2022quantum,zhu2022quantum}, and new classical simulation techniques~\cite{huang2020classical,barak2021spoofing,pan2022solving,oh2023spoofing,gao2024limitations} were designed to efficiently simulate the experiments on a classical computer, challenging whether advantage had been truly achieved.

For a convincing demonstration of quantum advantage, the task at hand should ideally satisfy the following three criteria: (i) it can be solved efficiently in a near-term quantum experiment; (ii) it is provably classically hard to solve, i.e., takes superpolynomial time as the system size scales; (iii) the solution can be efficiently verified with a classical computer with minimal trust in the experiment.

For random circuit sampling, the main problem is that the verification seems unavoidably tied to the underlying classical hardness. Consequently, it seems to inherently require exponential time to classically verify the output. A recent proposal to address this challenge involves the study of random peaked circuits~\cite{aaronson2024verifiable}; however, the current understanding of these circuits is insufficient to determine their viability as a quantum advantage experiment. There exist many other proposals, like Shor's algorithm~\cite{shor1994algorithms}, Yamakawa and Zhandry search~\cite{yamakawa2024verifiable}, variational quantum eigensolvers~\cite{Peruzzo_2014}, algorithms for verifying quantumness~\cite{brakerski2021cryptographic}, and more, that all seem to satisfy at most two of the above three points.

However, quantum resources offer more than just potential computational speedups, where ultimately time is the resource that one wants to minimize.  For example, it is known (or in same cases strongly conjectured) that superpolynomial advantages exist in communication complexity~\cite{buhrman2001quantum,gilboa2023exponential}, sample complexity~\cite{bshouty1995learning,grilo2019learning} and space complexity~\cite{gavinsky2007exponential,kallaugher2024exponential}. All these advantages stem from the fact that quantum information processing not only provides potential benefits in computation, but also in the way data can be stored, communicated and accessed in a quantum setting.

In this Letter, we present a new approach that combines the two ideas that quantum computers are inherently samplers, and that they can leverage quantum resources in ways classical computers cannot for their classical counterparts. We show that in a \textit{sample-to-sample} setting, quantum computing can achieve the largest possible advantage. Proofs and technical discussions are included in Supplemental Material (SM).

\vspace{.2cm}

\paragraph*{Quantum advantage in complement sampling ---}
Let $\omega = \{0,1\}^n$ be a set of $N=2^n$ elements. We will focus on the following problem.

\begin{problem}[Complement Sampling] Given (quantum) sampling access to the uniform distribution over a subset $S \subset \omega$ with cardinality $K$, output an element $y \in \bar{S}$ where $\bar{S}$ is the complement of $S$.
\end{problem}

Complement sampling can be viewed as a sample-to-sample problem where, given samples from some set $A$ according to a distribution $\mu_A$, one has to output something from a related set $B$, according to a distribution $\mu_B$.
This is different than, for example, most problems in classical~\cite{Haussler_1992, Hanneke_2016} and quantum~\cite{Aaronson_2007, arunachalam2017guest, Coopmans_2024} machine learning that use the number of samples as a complexity measure. For a concrete example, \cite{Arunachalam_2020} studies the sample complexity of the quantum coupon collector, where the task is to fully identify a subset $A$ given quantum sampling access to its elements. 
In contrast, complement sampling does not require learning anything about the subset $A$.

In this work, we study the power of quantum samples and quantum computing over classical samples in solving complement sampling. Given a probability distribution $D : S  \subseteq \omega\rightarrow [0,1]$, a quantum sample from $D$ is defined as having access to a single copy of the state~\cite{bshouty1995learning,arunachalam2017guest}
\begin{align*}
    \ket{S_D} = \sum_{x \in S} \sqrt{D(x)}  \ket{x}.
\end{align*}
When $D$ is the uniform distribution, we simply write $\ket{S}$. Measuring copies of $\ket{S_D}$ in the computational basis is identical to classical sampling from the distribution $D(x)$. Thus there would be no advantage to a classical algorithm in having access to the quantum samples if all the classical algorithm can do is computational basis measurements.

From a purely information-theoretic perspective, it is not immediately clear that quantum samples should provide any advantage to a quantum algorithm either. Since the required output for complement sampling is a classical $n$-bit string, it is natural to compare the amount of classical information received per sample. Holevo's theorem guarantees that the maximum information content per sample is the same for quantum and classical samples: both yield at most $n$ bits of classical information~\cite{holevo1973bounds}. Nevertheless, we prove that quantum samples can offer a dramatic advantage over classical samples in complement sampling: our main result is a characterization of the quantum and classical sample complexity in solving this task with certain success probabilities. 

\begin{theorem}[Informal, from Theorems 5 and 6 in SM] Let $1 \leq K \leq N-1$ be the cardinality of a subset $S \subset \omega$. Then there is a polynomial-time quantum algorithm that uses only $1$ quantum sample and solves complement sampling with probability 
\begin{align*}
    \frac{\min\{K,N-K\} }{\max\{K,N-K\}}.
\end{align*}
Let $\delta \in [0,\frac{1}{2}]$. Then any classical algorithm that solves complement sampling with success probability $\geq \frac{1}{2} + \delta$ needs at least
\begin{align*}
    N-\frac{2 (N-K)}{2 \delta + 1},
\end{align*}
distinct classical samples.
\label{thm:main_result_intro}
\end{theorem}

The classical lower bound is exactly matched by a simple randomized algorithm: it outputs a uniformly random element from the set of all possible elements, excluding those observed in the classical samples. The quantum algorithm, which we describe later, solves complement sampling by doing something even stronger than what is stated in Theorem~\ref{thm:main_result_intro}: it (approximately)
\begin{equation*}
    \ket{S} = \frac{1}{\sqrt{K}} \sum_{x \in S} \ket{x} \; \longrightarrow \; \ket{\bar{S}} = \frac{1}{\sqrt{N - K}} \sum_{x \notin S} \ket{x} ,
\end{equation*}
from which a sample from $\bar{S}$ is simply obtained by performing a measurement in the computational basis.

Hereafter, we use Landau symbols $\mathcal{O}$ and $\Omega$ for asymptotic upper and lower bounds, respectively. For cases where polylogarithmic terms are omitted, we use $\tilde{\mathcal{O}}$ and $\tilde{\Omega}$ instead. Theorem~\ref{thm:main_result_intro} shows that for $K=N/2$, the classical sample complexity is $\Omega(N)$ for any $\delta = \Omega(1)$ whilst the quantum algorithm succeeds with probability $1$ using only a single quantum sample. From a theoretical perspective, this is a remarkable feat: it shows that in a sample-to-sample setting, quantum computation can achieve the largest possible separation from classical computation, namely a gap of $1$ versus a linear (in $N$) number of samples. This contrasts with, for example, the case of black-box query complexity, as Aaronson and Ambainis~\cite{aaronson2015forrelation} showed that no partial Boolean function has quantum query complexity $1$ while requiring a linear number of randomized queries~\footnote{They proved that the separation can be at most $1$ versus $\tilde{\Omega}(\sqrt{N})$, which is achieved by the Forrelation problem. For a constant number of queries, one can achieve $\mathcal{O}_\epsilon(1)$ versus $\Omega(N^{1-\epsilon})$ for any $\epsilon >0$~\cite{bansal2021k}. Here $\mathcal{O}_{\epsilon}(f(n))$ means $\mathcal{O}(f(n))$ where the hidden constant is allowed to depend on $\epsilon$, which is treated as a fixed parameter.}

Using Kolmogorov complexity we also show an explicit circuit lower bound on any classical sampler that uses only a polynomial number of samples and returns samples uniformly at random from the complement (which the quantum algorithm allows one to do).

\begin{theorem}[Informal, from Theorem 8 in SM] Let $K=N/2$ and $N=2^n$. For $n \geq 1$ sufficiently large, any sampler that takes as input $\poly(n)$ classical samples from $S$ and outputs samples from $\bar{S}$ uniformly at random must have circuit complexity at least $\tilde{\Omega}(N)$.
\end{theorem}

The Kolmogorov complexity argument is particularly insightful of why complement sampling works in a quantum setting, as it highlights the key difference between classical and quantum samples. Classical samples can be cloned and reused indefinitely. For example, they can be used to construct a circuit that describes $S$ and, in turns, a sampler for $\bar{S}$. In comparison, our quantum algorithm \emph{destroys} the quantum sample in the process of providing a sample from the complement. This means that the quantum algorithm needs a new input for every generated output. 

\vspace{.2cm}

\paragraph*{The quantum algorithm ---}
Our initial inspiration comes from the construction by Aaronson, Atia and Susskind~\cite{aaronson2020hardness}. The authors show how several notions of complexity on quantum states are related in intricate ways. Consider two arbitrary orthogonal states $\ket{a}$ and $\ket{b}$ spanning a two-dimensional space. The conjugate (or Fourier) basis within this subspace is given by the equal superpositions $\ket{\phi^+} = \left(\ket{a}  + \ket{b}\right)/\sqrt{2}$ and $\ket{\phi^-} = \left(\ket{a}  - \ket{b}\right)/\sqrt{2}$. 
Their main result is that a circuit which (approximately) swaps $\ket{a}$ to $\ket{b}$ and back can be used to build another circuit which (approximately) distinguishes between $\ket{\phi^+}$ and $\ket{\phi^-}$. Vice versa, a distinguisher for $\ket{\phi^+}$ and $\ket{\phi^-}$ can be used to make a swapper for $\ket{a}$ and $\ket{b}$. This implies that the \emph{swapping complexity} of two orthogonal states is related to the \emph{distinguishing complexity} of the conjugate states. In particular, an efficient implementation of one implies an efficient implementation of the other.

How can this be used for the complement sampling problem? A Boolean function $f$ on the domain $\omega$ defines the phase state 
\begin{align*}
     \ket{y_{f}} = \frac{1}{\sqrt{N}} \sum_{x \in \omega} (-1)^{f(x)} \ket{x} .
\end{align*}
Consider the case of subsets of cardinality $K=N/2$. It can be observed that any subset state and its complementary state can be expressed using phase states
\begin{align*}
    \ket{S} &= \frac{1}{\sqrt{2}} \left( \ket{y_{f_\mathrm{con}}} + \ket{y_{f_\mathrm{bal}}} \right) ,\\ 
    \ket{\bar{S}} &= \frac{1}{\sqrt{2}} \left( \ket{y_{f_\mathrm{con}}} - \ket{y_{f_\mathrm{bal}}} \right),
\end{align*}
where $f_\mathrm{con}$ is a constant function (it maps all inputs $x \in \omega$ to $0$) and $f_\mathrm{bal}$ is a balanced function (half of the inputs in $\omega$ get mapped to $0$ and the other half get mapped to $1$). As the two states are orthogonal, the conjugate states are
\begin{align*}
    \ket{y_{f_\mathrm{con}}} &= \frac{1}{\sqrt{2}} \left( \ket{S} + \ket{\bar{S}} \right), \\ 
    \ket{y_{f_\mathrm{bal}}} &= \frac{1}{\sqrt{2}} \left( \ket{S} - \ket{\bar{S}} \right).
\end{align*}
But these are just the postquery states of the Deutsch-Jozsa algorithm~\cite{deutsch1992rapid}, and therefore can be efficiently distinguished with probability $1$. Figure~\ref{fig:circuits}(a) illustrates the distinguisher circuit. Using the Aaronson, Atia and Susskind construction, the Deutsch-Jozsa distinguisher can be turned into a swapper for $\ket{S}$ and $\ket{\bar{S}}$. The swapper is illustrated in Fig.~\ref{fig:circuits}(b), and derived in Sec.~I of SM. We can now argue that no perfect swapper for $\ket{S}$ and $\ket{\bar{S}}$ exists for values $K \neq N/2$. Indeed, only for $K=N/2$ we have that any two different subsets $S_1$, $S_2$ have perfectly distinguishable orthogonal conjugate states (see Proposition 2 in SM). If a perfect swapper would exist for $K \neq N/2$, one could violate the distinguishability limits given by the Helstrom bound~\cite{helstrom1969quantum}.

\begin{figure}
    \centering
    \includegraphics[width=\linewidth]{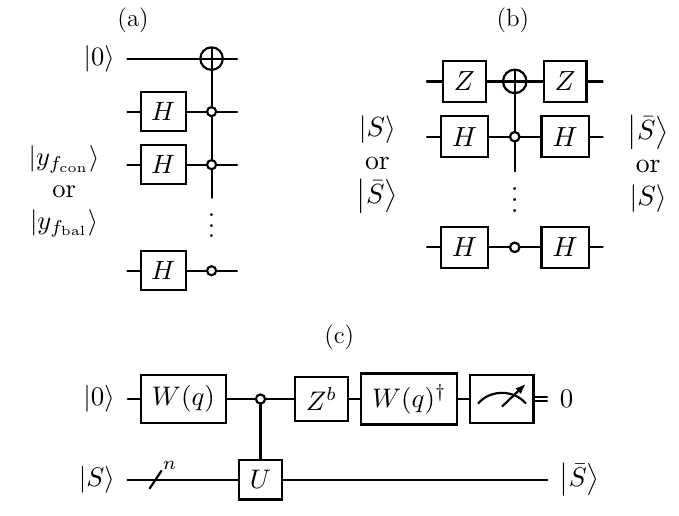}
    \caption{Distinguisher and swappers for complement sampling. (a) Our distinguisher can be thought of as the Deutsch-Jozsa algorithm followed by a copy of the measurement outcome into an ancilla qubit. (b) The complement swapper can be thought of as the Grover diffusion operator $U = 2 \dyad{+^n} - \mathbb{I}$ up to a global phase. (c) The zero-error swapper uses the controlled-$U$, the gate $W(q) = e^{i \arccos(\sqrt{q}) Y}$, and an ancilla qubit that upon measurement outcome 0 indicates correctness. The parameters are set to $b = 0$, $q = \frac{1}{2(1 - K/N)}$ if $K < \frac{N}{2}$, and to $b = 1$, $q = \frac{N}{2K}$ if $K \geq \frac{N}{2}$.}
\label{fig:circuits}
\end{figure}

By inspecting the swapper circuit based on the Deutsch-Jozsa distinguisher, we find that, restricted to the register where $\ket{S}$ is held, it implements a reflection around the uniform superposition of basis states. Up to a global phase, this is the Grover diffusion operator $U = 2 \dyad{+^n} - \mathbb{I}$. With this new insight we generalize the algorithm to $K \neq N/2$ in two different ways. The first, which we simply call \emph{complement swapper}, directly applies $U$ to any $\ket{S}$. For values of $K$ a constant fraction away from $N/2$ the resulting state has  constant overlap with the complement $\ket{\bar{S}}$. The second algorithm, which we call \emph{zero-error swapper}, uses an additional ancilla as a flag qubit and a controlled application of $U$. This is illustrated in Fig.~\ref{fig:circuits}(c).

Figure~\ref{fig:hitting_rate} compares the performance of four algorithms as a function of $\beta \in [-1/2, 1/2]$, which we define as the deviation from the ideal subset size ratio $K/N = 1/2 + \beta$. The complement swapper outputs complementary samples with probability $1 - 4\beta^2$ (solid red line), but does not provide information on the correctness of the sample. The zero-error swapper outputs complementary samples with probability $\frac{1 - 2|\beta|}{1 + 2|\beta|}$ (dashed blue line), with the guarantee that the output is correct (i.e., it is a \emph{Las Vegas} algorithm). Theorem 5 in SM shows that this is the optimal success probability for any zero-error algorithm when $K \geq N/2$, as well as when no extra auxiliary qubits are allowed (except for the flag qubit) and $K < N/2$. In our proof we construct the hardest instances of $S$ in the sense of maximizing the violation of the conservation of inner products under unitary transformations, which puts an upper bound on the achievable success probability. The quantum algorithm for the coupon collector~\cite{Arunachalam_2020} can also be used for complement sampling (purple dash-dotted line), but it achieves half the probability of success of our complement swapper. Our improvement arises from the coherent reflection around the uniform superposition, whereas the quantum coupon collector performs a measurement $\{\ketbra{+^n},\mathbb{I}-\ketbra{+^n}\}$ on $\ket{S}$, followed by a measurement in the computational basis. The authors~\cite{Arunachalam_2020} additionally consider a stronger input model which gives access to the state preparation circuit for $\ket{S}$. We do not compare against this version of their algorithm because complement sampling does not allow access to the state preparation circuit. Finally, the success probability of the optimal classical algorithm based on random guessing is $\frac{1/2 - \beta}{1 - 1/N} \approx \frac{1}{2} - \beta$ (green dotted line), which is large for negative deviations.

\begin{figure}
    \centering
    \includegraphics[width=\linewidth]{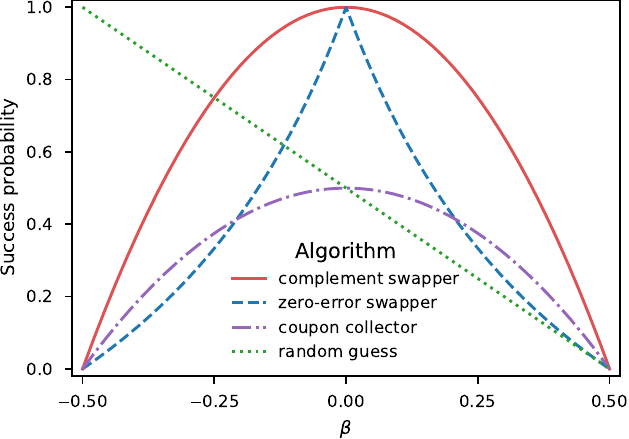}
    \caption{Comparison of complement sampling algorithms given a single sample as input. The vertical axis is the probability of obtaining a sample from the complementary set. The horizontal axis is the deviation from the ideal subset size ratio $K/N = 1/2 + \beta$. When $\beta=0$ the quantum algorithms based on swappers are exact.}
    \label{fig:hitting_rate}
\end{figure}

\vspace{.2cm}

\paragraph*{Towards experimental demonstrations ---}
So far we have only discussed worst-case classical hardness results, which are problematic if one also wants to instantiate the problem in an actual experiment. Even worse, one might argue that these hard instances hide all their complexity in the state $\ket{S}$, a comment that we justify in Sec.~III of SM by proving that most states of the form $\ket{S}$ with $K=N/2$ have exponential quantum circuit complexity~\footnote{This circuit lower bound was to be expected under the assumption that $\QMA \neq \QCMA$, because if all subset states could be (approximately) constructed with a polynomially-sized circuit then we would have that $\QCMA= \QMA$ by the result of~\cite{grilo2015qma}.}. To show that complement sampling is indeed capable of demonstrating quantum advantage, we show that the problem remains hard for most choices of $S$, as long as there is no particular structure to them. We focus on the case where $K=N/2$. Consider the family of subsets $\mathcal{S} = \{ S^k\}_k$ where each subset is given by $S^k= \{y : y= P^k (x) \text{ and the first bit of } x\in \omega \text{ is } 0 \}$ for some strong pseudorandom permutation $P^k$ with a key $k \in \{0,1\}^{\lambda}$, and $\lambda \in \mathbb{N}$ is the security parameter. We prove the following result.

\begin{theorem}[Informal, from Theorem 7 in SM] Let $K=N/2$ and $N=2^n$. Let $A$ be a polynomial time algorithm that has sampling access to $S$. Under the assumption of the existence of one-way functions, for all such $A$ it must hold that
\begin{align*}
    \Pr \left[A \textup{ outputs a } y \in \bar{S}\right] < \frac{1}{2} + \frac{1}{\poly(n)}.
\end{align*}
\label{thm:intro_sPRP}
\end{theorem}

Equipped with the above results, we propose an experiment to demonstrate the advantage of quantum over classical resources in a \emph{provable}, \emph{verifiable} and \emph{NISQable} manner.
For a total of $r$ rounds, a player and referee play the following interactive game: (1) The referee picks a pseudorandom permutation $P$ with inverse $P^{-1}$ and defines $S = \{y : y= P(x)$ and the first bit of $x\in \omega$ is $0\}$. (2) The (quantum) player asks for $j$ (quantum) samples corresponding to the uniform distribution over elements in $S$. A quantum player can perform coherent operations on the quantum samples. The player sends back a candidate element $\hat{y}$ supposedly from $\bar{S}$. (3) The referee verifies that $\hat{y} \in \bar{S}$ by checking if $\hat{x} = P^{-1}(\hat{y})$ has the first bit equal to $1$.

A sketch of the proposed experiment is provided in Fig.~\ref{fig:CS_summarized}. The best strategy for the classical player is to output any sample different from the observed samples (we prove this in Sec.~II of SM). For $j=1$, this random guessing strategy has success probability of $\tfrac{2^{n-1}}{2^n-1}$, quickly approaching $1/2$ for large $n$. By repeating the experiment a total of $r$ times, each time with a different pseudorandom permutation, the probability of succeeding in all rounds is approximately $1/2^r$. On the other hand, the quantum player can use our algorithm to obtain $\ket{\bar{S}}$ from $\ket{S}$. This strategy succeeds with probability $1$ at each of the $r$ rounds. In reality, the quantum player only needs to succeed with probability $\geq \frac{1}{2} + 1/\poly(n)$ at each round in order to beat the classical player over $r = \poly(n)$ rounds. Thus, the player can tolerate a small rate of noise in the quantum computation. For example, in the simplistic noise model where each gate fails with probability $\lambda$, the overall circuit fidelity is of the order of $(1-\lambda)^m$ where $m$ is the number of gates. As we will see, our quantum circuit only needs $\mathcal{O}(n)$ gates and $\mathcal{O}(\log n)$ circuit depth; the bottleneck being the $n$-qubit Toffoli gate, which can be implemented in $\mathcal{O}(\log n)$ circuit depth with $1$ additional ancilla using the construction from~\cite{nie2024quantum}. 

The preparation of the state $\ket{S}$ by the referee is computationally more expensive as it requires the implementation of pseudorandom permutations. As a near-term toy example, one could implement permutations using the Simplified Advanced Encryption Standard (S-AES) protocol. S-AES uses a block size of $16$ bits (as opposed to $128$ bits for AES) and a key size of $16$ bits as well ($128$, $192$ or $256$ for AES). A system size of $16$ can display features of quantum advantage: from Theorem~\ref{thm:main_result_intro} with, e.g.,  $N=2^{16}, K=2^{15}$ and $\delta=1/6$, any classical algorithm solving complement sampling requires $2^{14} = 16384$ distinct samples. A quantum circuit implementation of S-AES based on~\cite{wang2022quantum} requires only $48$ qubits, $168$ Toffoli gates, $364$ controlled NOT gates, and $75$ NOT gates. Beyond $n=16$ one could resort to \emph{approximate} pseudorandom permutations built from layers of random reversible 3-bit gates. Exponentially small error can be achieved in depth linear in $n$ when using one-dimensional nearest-neighbor layers, and in depth $\sqrt{n}$ when using two-dimensional layers~\cite{gay2025pseudorandomness}. When noise affects the computation, the referee is allowed to use quantum error detection codes~\cite{knill2004postselected} before sending the state $\ket{S}$ to the player. By discarding circuit runs where an error is detected, the referee is expected to increase the fidelity of the state preparation. This comes at the cost of introducing more qubits, gates, and runs, but near-term protocols have shown promising results with low overheads~\cite{self2024protecting}.

We expect that in the first complement sampling experiments all computations from both referee and player will be performed on the same device. Smart choices of families of subsets, and the use of circuit optimizations and error detection, may lead to a NISQable advantage experiment. 

\begin{figure}
    \centering
    \includegraphics[width=\linewidth]{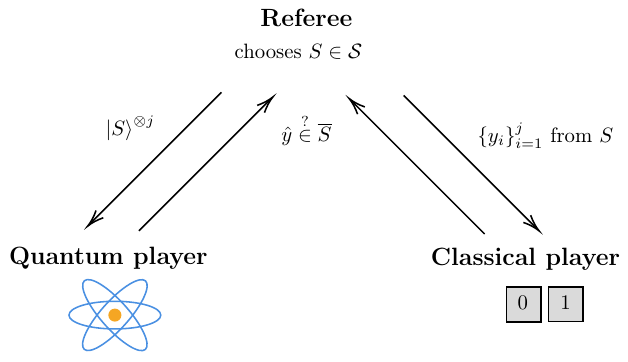}
    \caption{A single round of the complement sampling game with $j$ input samples. The referee picks subset $S$ from a family of subsets of fixed cardinality $\mathcal{S}$. The quantum player gets access to $j$ copies of the state $\ket{S}$, and the classical player gets access to $j$ elements sampled according to the uniform distribution over $S$. Giving the classical player access to measurements of $\ket{S}$ in the computational basis would yield the same input model. Giving the quantum player access to classical samples does not yield any advantage over the classical player. Each player sends back a candidate element $\hat{y}$ and the referee checks whether it belongs to the complementary subset $\bar{S}$.}
    \label{fig:CS_summarized}
\end{figure}

\vspace{.2cm}

\paragraph*{Discussion and future work ---}
In this work we have shown that in a sample-to-sample setting quantum algorithms provide the largest possible advantage over their classical counterparts. The classical hardness of the problem is very robust: even instances that are generated using strong pseudorandom permutations must be classically hard under the assumption of the existence of one-way functions. On the other hand, the quantum advantage vanishes if classical samples are provided as input instead of the quantum samples. This indicates that the present approach does not provide a purely computational quantum advantage, where both the input and output to a problem are classical.

Our work opens up a new direction in demonstrating an advantage of quantum over classical resources. Our back-of-the-envelope resource estimation shows that a proof-of-concept demonstration of complement sampling is NISQable, but further circuit optimizations are likely needed to solve instances that are truly classically intractable. Progress on this matter was recently reported in Ref.~\cite{benedetti2025unconditional}, where a refined version of the complement sampling game was successfully executed on quantum hardware.

Another topic for future work is the search for other applications of complement sampling. In particular, in the field of cryptography, we believe that there exist applications that can utilize the power of quantum computing for complement sampling.

\vspace{.2cm}

\paragraph*{Acknowledgments ---}
We thank Florian Curchod and Pranav Kalidindi for providing feedback on an earlier version of this paper. J.W. was supported by the Dutch Ministry of Economic Affairs and Climate Policy (EZK), as part of the Quantum Delta NL program.

\vspace{.2cm}

\let\oldaddcontentsline\addcontentsline
\renewcommand{\addcontentsline}[3]{}
\bibliography{bibliography}

\begin{thebibliography}{48}%
\makeatletter
\providecommand \@ifxundefined [1]{%
 \@ifx{#1\undefined}
}%
\providecommand \@ifnum [1]{%
 \ifnum #1\expandafter \@firstoftwo
 \else \expandafter \@secondoftwo
 \fi
}%
\providecommand \@ifx [1]{%
 \ifx #1\expandafter \@firstoftwo
 \else \expandafter \@secondoftwo
 \fi
}%
\providecommand \natexlab [1]{#1}%
\providecommand \enquote  [1]{``#1''}%
\providecommand \bibnamefont  [1]{#1}%
\providecommand \bibfnamefont [1]{#1}%
\providecommand \citenamefont [1]{#1}%
\providecommand \href@noop [0]{\@secondoftwo}%
\providecommand \href [0]{\begingroup \@sanitize@url \@href}%
\providecommand \@href[1]{\@@startlink{#1}\@@href}%
\providecommand \@@href[1]{\endgroup#1\@@endlink}%
\providecommand \@sanitize@url [0]{\catcode `\\12\catcode `\$12\catcode `\&12\catcode `\#12\catcode `\^12\catcode `\_12\catcode `\%12\relax}%
\providecommand \@@startlink[1]{}%
\providecommand \@@endlink[0]{}%
\providecommand \url  [0]{\begingroup\@sanitize@url \@url }%
\providecommand \@url [1]{\endgroup\@href {#1}{\urlprefix }}%
\providecommand \urlprefix  [0]{URL }%
\providecommand \Eprint [0]{\href }%
\providecommand \doibase [0]{https://doi.org/}%
\providecommand \selectlanguage [0]{\@gobble}%
\providecommand \bibinfo  [0]{\@secondoftwo}%
\providecommand \bibfield  [0]{\@secondoftwo}%
\providecommand \translation [1]{[#1]}%
\providecommand \BibitemOpen [0]{}%
\providecommand \bibitemStop [0]{}%
\providecommand \bibitemNoStop [0]{.\EOS\space}%
\providecommand \EOS [0]{\spacefactor3000\relax}%
\providecommand \BibitemShut  [1]{\csname bibitem#1\endcsname}%
\let\auto@bib@innerbib\@empty
\bibitem [{\citenamefont {Arute}\ \emph {et~al.}(2019)\citenamefont {Arute}, \citenamefont {Arya}, \citenamefont {Babbush}, \citenamefont {Bacon}, \citenamefont {Bardin}, \citenamefont {Barends}, \citenamefont {Biswas}, \citenamefont {Boixo}, \citenamefont {Brandao}, \citenamefont {Buell} \emph {et~al.}}]{arute2019quantum}%
  \BibitemOpen
  \bibfield  {author} {\bibinfo {author} {\bibfnamefont {F.}~\bibnamefont {Arute}}, \bibinfo {author} {\bibfnamefont {K.}~\bibnamefont {Arya}}, \bibinfo {author} {\bibfnamefont {R.}~\bibnamefont {Babbush}}, \bibinfo {author} {\bibfnamefont {D.}~\bibnamefont {Bacon}}, \bibinfo {author} {\bibfnamefont {J.~C.}\ \bibnamefont {Bardin}}, \bibinfo {author} {\bibfnamefont {R.}~\bibnamefont {Barends}}, \bibinfo {author} {\bibfnamefont {R.}~\bibnamefont {Biswas}}, \bibinfo {author} {\bibfnamefont {S.}~\bibnamefont {Boixo}}, \bibinfo {author} {\bibfnamefont {F.~G.}\ \bibnamefont {Brandao}}, \bibinfo {author} {\bibfnamefont {D.~A.}\ \bibnamefont {Buell}}, \emph {et~al.},\ }\bibfield  {title} {\bibinfo {title} {Quantum supremacy using a programmable superconducting processor},\ }\href@noop {} {\bibfield  {journal} {\bibinfo  {journal} {Nature}\ }\textbf {\bibinfo {volume} {574}},\ \bibinfo {pages} {505} (\bibinfo {year} {2019})}\BibitemShut {NoStop}%
\bibitem [{\citenamefont {Zhong}\ \emph {et~al.}(2020)\citenamefont {Zhong}, \citenamefont {Wang}, \citenamefont {Deng}, \citenamefont {Chen}, \citenamefont {Peng}, \citenamefont {Luo}, \citenamefont {Qin}, \citenamefont {Wu}, \citenamefont {Ding}, \citenamefont {Hu} \emph {et~al.}}]{zhong2020quantum}%
  \BibitemOpen
  \bibfield  {author} {\bibinfo {author} {\bibfnamefont {H.-S.}\ \bibnamefont {Zhong}}, \bibinfo {author} {\bibfnamefont {H.}~\bibnamefont {Wang}}, \bibinfo {author} {\bibfnamefont {Y.-H.}\ \bibnamefont {Deng}}, \bibinfo {author} {\bibfnamefont {M.-C.}\ \bibnamefont {Chen}}, \bibinfo {author} {\bibfnamefont {L.-C.}\ \bibnamefont {Peng}}, \bibinfo {author} {\bibfnamefont {Y.-H.}\ \bibnamefont {Luo}}, \bibinfo {author} {\bibfnamefont {J.}~\bibnamefont {Qin}}, \bibinfo {author} {\bibfnamefont {D.}~\bibnamefont {Wu}}, \bibinfo {author} {\bibfnamefont {X.}~\bibnamefont {Ding}}, \bibinfo {author} {\bibfnamefont {Y.}~\bibnamefont {Hu}}, \emph {et~al.},\ }\bibfield  {title} {\bibinfo {title} {Quantum computational advantage using photons},\ }\href@noop {} {\bibfield  {journal} {\bibinfo  {journal} {Science}\ }\textbf {\bibinfo {volume} {370}},\ \bibinfo {pages} {1460} (\bibinfo {year} {2020})}\BibitemShut {NoStop}%
\bibitem [{\citenamefont {Madsen}\ \emph {et~al.}(2022)\citenamefont {Madsen}, \citenamefont {Laudenbach}, \citenamefont {Askarani}, \citenamefont {Rortais}, \citenamefont {Vincent}, \citenamefont {Bulmer}, \citenamefont {Miatto}, \citenamefont {Neuhaus}, \citenamefont {Helt}, \citenamefont {Collins} \emph {et~al.}}]{madsen2022quantum}%
  \BibitemOpen
  \bibfield  {author} {\bibinfo {author} {\bibfnamefont {L.~S.}\ \bibnamefont {Madsen}}, \bibinfo {author} {\bibfnamefont {F.}~\bibnamefont {Laudenbach}}, \bibinfo {author} {\bibfnamefont {M.~F.}\ \bibnamefont {Askarani}}, \bibinfo {author} {\bibfnamefont {F.}~\bibnamefont {Rortais}}, \bibinfo {author} {\bibfnamefont {T.}~\bibnamefont {Vincent}}, \bibinfo {author} {\bibfnamefont {J.~F.}\ \bibnamefont {Bulmer}}, \bibinfo {author} {\bibfnamefont {F.~M.}\ \bibnamefont {Miatto}}, \bibinfo {author} {\bibfnamefont {L.}~\bibnamefont {Neuhaus}}, \bibinfo {author} {\bibfnamefont {L.~G.}\ \bibnamefont {Helt}}, \bibinfo {author} {\bibfnamefont {M.~J.}\ \bibnamefont {Collins}}, \emph {et~al.},\ }\bibfield  {title} {\bibinfo {title} {Quantum computational advantage with a programmable photonic processor},\ }\href@noop {} {\bibfield  {journal} {\bibinfo  {journal} {Nature}\ }\textbf {\bibinfo {volume} {606}},\ \bibinfo {pages} {75} (\bibinfo {year} {2022})}\BibitemShut {NoStop}%
\bibitem [{\citenamefont {Zhu}\ \emph {et~al.}(2022)\citenamefont {Zhu}, \citenamefont {Cao}, \citenamefont {Chen}, \citenamefont {Chen}, \citenamefont {Chen}, \citenamefont {Chung}, \citenamefont {Deng}, \citenamefont {Du}, \citenamefont {Fan}, \citenamefont {Gong} \emph {et~al.}}]{zhu2022quantum}%
  \BibitemOpen
  \bibfield  {author} {\bibinfo {author} {\bibfnamefont {Q.}~\bibnamefont {Zhu}}, \bibinfo {author} {\bibfnamefont {S.}~\bibnamefont {Cao}}, \bibinfo {author} {\bibfnamefont {F.}~\bibnamefont {Chen}}, \bibinfo {author} {\bibfnamefont {M.-C.}\ \bibnamefont {Chen}}, \bibinfo {author} {\bibfnamefont {X.}~\bibnamefont {Chen}}, \bibinfo {author} {\bibfnamefont {T.-H.}\ \bibnamefont {Chung}}, \bibinfo {author} {\bibfnamefont {H.}~\bibnamefont {Deng}}, \bibinfo {author} {\bibfnamefont {Y.}~\bibnamefont {Du}}, \bibinfo {author} {\bibfnamefont {D.}~\bibnamefont {Fan}}, \bibinfo {author} {\bibfnamefont {M.}~\bibnamefont {Gong}}, \emph {et~al.},\ }\bibfield  {title} {\bibinfo {title} {Quantum computational advantage via 60-qubit 24-cycle random circuit sampling},\ }\href@noop {} {\bibfield  {journal} {\bibinfo  {journal} {Science bulletin}\ }\textbf {\bibinfo {volume} {67}},\ \bibinfo {pages} {240} (\bibinfo {year} {2022})}\BibitemShut {NoStop}%
\bibitem [{\citenamefont {Huang}\ \emph {et~al.}(2020)\citenamefont {Huang}, \citenamefont {Zhang}, \citenamefont {Newman}, \citenamefont {Cai}, \citenamefont {Gao}, \citenamefont {Tian}, \citenamefont {Wu}, \citenamefont {Xu}, \citenamefont {Yu}, \citenamefont {Yuan} \emph {et~al.}}]{huang2020classical}%
  \BibitemOpen
  \bibfield  {author} {\bibinfo {author} {\bibfnamefont {C.}~\bibnamefont {Huang}}, \bibinfo {author} {\bibfnamefont {F.}~\bibnamefont {Zhang}}, \bibinfo {author} {\bibfnamefont {M.}~\bibnamefont {Newman}}, \bibinfo {author} {\bibfnamefont {J.}~\bibnamefont {Cai}}, \bibinfo {author} {\bibfnamefont {X.}~\bibnamefont {Gao}}, \bibinfo {author} {\bibfnamefont {Z.}~\bibnamefont {Tian}}, \bibinfo {author} {\bibfnamefont {J.}~\bibnamefont {Wu}}, \bibinfo {author} {\bibfnamefont {H.}~\bibnamefont {Xu}}, \bibinfo {author} {\bibfnamefont {H.}~\bibnamefont {Yu}}, \bibinfo {author} {\bibfnamefont {B.}~\bibnamefont {Yuan}}, \emph {et~al.},\ }\bibfield  {title} {\bibinfo {title} {Classical simulation of quantum supremacy circuits},\ }\href@noop {} {\bibfield  {journal} {\bibinfo  {journal} {arXiv preprint arXiv:2005.06787}\ } (\bibinfo {year} {2020})}\BibitemShut {NoStop}%
\bibitem [{\citenamefont {Barak}\ \emph {et~al.}(2021)\citenamefont {Barak}, \citenamefont {Chou},\ and\ \citenamefont {Gao}}]{barak2021spoofing}%
  \BibitemOpen
  \bibfield  {author} {\bibinfo {author} {\bibfnamefont {B.}~\bibnamefont {Barak}}, \bibinfo {author} {\bibfnamefont {C.-N.}\ \bibnamefont {Chou}},\ and\ \bibinfo {author} {\bibfnamefont {X.}~\bibnamefont {Gao}},\ }\bibfield  {title} {\bibinfo {title} {Spoofing linear cross-entropy benchmarking in shallow quantum circuits},\ }in\ \href@noop {} {\emph {\bibinfo {booktitle} {12th Innovations in Theoretical Computer Science Conference (ITCS 2021)}}},\ Vol.\ \bibinfo {volume} {185}\ (\bibinfo {year} {2021})\ p.~\bibinfo {pages} {30}\BibitemShut {NoStop}%
\bibitem [{\citenamefont {Pan}\ \emph {et~al.}(2022)\citenamefont {Pan}, \citenamefont {Chen},\ and\ \citenamefont {Zhang}}]{pan2022solving}%
  \BibitemOpen
  \bibfield  {author} {\bibinfo {author} {\bibfnamefont {F.}~\bibnamefont {Pan}}, \bibinfo {author} {\bibfnamefont {K.}~\bibnamefont {Chen}},\ and\ \bibinfo {author} {\bibfnamefont {P.}~\bibnamefont {Zhang}},\ }\bibfield  {title} {\bibinfo {title} {Solving the sampling problem of the sycamore quantum circuits},\ }\href@noop {} {\bibfield  {journal} {\bibinfo  {journal} {Phys. Rev. Lett.}\ }\textbf {\bibinfo {volume} {129}},\ \bibinfo {pages} {090502} (\bibinfo {year} {2022})}\BibitemShut {NoStop}%
\bibitem [{\citenamefont {Oh}\ \emph {et~al.}(2023)\citenamefont {Oh}, \citenamefont {Jiang},\ and\ \citenamefont {Fefferman}}]{oh2023spoofing}%
  \BibitemOpen
  \bibfield  {author} {\bibinfo {author} {\bibfnamefont {C.}~\bibnamefont {Oh}}, \bibinfo {author} {\bibfnamefont {L.}~\bibnamefont {Jiang}},\ and\ \bibinfo {author} {\bibfnamefont {B.}~\bibnamefont {Fefferman}},\ }\bibfield  {title} {\bibinfo {title} {Spoofing cross-entropy measure in boson sampling},\ }\href@noop {} {\bibfield  {journal} {\bibinfo  {journal} {Phys. Rev. Lett.}\ }\textbf {\bibinfo {volume} {131}},\ \bibinfo {pages} {010401} (\bibinfo {year} {2023})}\BibitemShut {NoStop}%
\bibitem [{\citenamefont {Gao}\ \emph {et~al.}(2024)\citenamefont {Gao}, \citenamefont {Kalinowski}, \citenamefont {Chou}, \citenamefont {Lukin}, \citenamefont {Barak},\ and\ \citenamefont {Choi}}]{gao2024limitations}%
  \BibitemOpen
  \bibfield  {author} {\bibinfo {author} {\bibfnamefont {X.}~\bibnamefont {Gao}}, \bibinfo {author} {\bibfnamefont {M.}~\bibnamefont {Kalinowski}}, \bibinfo {author} {\bibfnamefont {C.-N.}\ \bibnamefont {Chou}}, \bibinfo {author} {\bibfnamefont {M.~D.}\ \bibnamefont {Lukin}}, \bibinfo {author} {\bibfnamefont {B.}~\bibnamefont {Barak}},\ and\ \bibinfo {author} {\bibfnamefont {S.}~\bibnamefont {Choi}},\ }\bibfield  {title} {\bibinfo {title} {Limitations of linear cross-entropy as a measure for quantum advantage},\ }\href@noop {} {\bibfield  {journal} {\bibinfo  {journal} {PRX Quantum}\ }\textbf {\bibinfo {volume} {5}},\ \bibinfo {pages} {010334} (\bibinfo {year} {2024})}\BibitemShut {NoStop}%
\bibitem [{\citenamefont {Aaronson}\ and\ \citenamefont {Zhang}(2024)}]{aaronson2024verifiable}%
  \BibitemOpen
  \bibfield  {author} {\bibinfo {author} {\bibfnamefont {S.}~\bibnamefont {Aaronson}}\ and\ \bibinfo {author} {\bibfnamefont {Y.}~\bibnamefont {Zhang}},\ }\bibfield  {title} {\bibinfo {title} {On verifiable quantum advantage with peaked circuit sampling},\ }\href@noop {} {\bibfield  {journal} {\bibinfo  {journal} {arXiv preprint arXiv:2404.14493}\ } (\bibinfo {year} {2024})}\BibitemShut {NoStop}%
\bibitem [{\citenamefont {Shor}(1994)}]{shor1994algorithms}%
  \BibitemOpen
  \bibfield  {author} {\bibinfo {author} {\bibfnamefont {P.~W.}\ \bibnamefont {Shor}},\ }\bibfield  {title} {\bibinfo {title} {Algorithms for quantum computation: discrete logarithms and factoring},\ }in\ \href@noop {} {\emph {\bibinfo {booktitle} {Proceedings 35th annual symposium on foundations of computer science}}}\ (\bibinfo {organization} {Ieee},\ \bibinfo {year} {1994})\ pp.\ \bibinfo {pages} {124--134}\BibitemShut {NoStop}%
\bibitem [{\citenamefont {Yamakawa}\ and\ \citenamefont {Zhandry}(2024)}]{yamakawa2024verifiable}%
  \BibitemOpen
  \bibfield  {author} {\bibinfo {author} {\bibfnamefont {T.}~\bibnamefont {Yamakawa}}\ and\ \bibinfo {author} {\bibfnamefont {M.}~\bibnamefont {Zhandry}},\ }\bibfield  {title} {\bibinfo {title} {Verifiable quantum advantage without structure},\ }\href@noop {} {\bibfield  {journal} {\bibinfo  {journal} {Journal of the ACM}\ }\textbf {\bibinfo {volume} {71}},\ \bibinfo {pages} {1} (\bibinfo {year} {2024})}\BibitemShut {NoStop}%
\bibitem [{\citenamefont {Peruzzo}\ \emph {et~al.}(2014)\citenamefont {Peruzzo}, \citenamefont {McClean}, \citenamefont {Shadbolt}, \citenamefont {Yung}, \citenamefont {Zhou}, \citenamefont {Love}, \citenamefont {Aspuru-Guzik},\ and\ \citenamefont {O’Brien}}]{Peruzzo_2014}%
  \BibitemOpen
  \bibfield  {author} {\bibinfo {author} {\bibfnamefont {A.}~\bibnamefont {Peruzzo}}, \bibinfo {author} {\bibfnamefont {J.}~\bibnamefont {McClean}}, \bibinfo {author} {\bibfnamefont {P.}~\bibnamefont {Shadbolt}}, \bibinfo {author} {\bibfnamefont {M.-H.}\ \bibnamefont {Yung}}, \bibinfo {author} {\bibfnamefont {X.-Q.}\ \bibnamefont {Zhou}}, \bibinfo {author} {\bibfnamefont {P.~J.}\ \bibnamefont {Love}}, \bibinfo {author} {\bibfnamefont {A.}~\bibnamefont {Aspuru-Guzik}},\ and\ \bibinfo {author} {\bibfnamefont {J.~L.}\ \bibnamefont {O’Brien}},\ }\bibfield  {title} {\bibinfo {title} {A variational eigenvalue solver on a photonic quantum processor},\ }\href@noop {} {\bibfield  {journal} {\bibinfo  {journal} {Nature Communications}\ }\textbf {\bibinfo {volume} {5}} (\bibinfo {year} {2014})}\BibitemShut {NoStop}%
\bibitem [{\citenamefont {Brakerski}\ \emph {et~al.}(2021)\citenamefont {Brakerski}, \citenamefont {Christiano}, \citenamefont {Mahadev}, \citenamefont {Vazirani},\ and\ \citenamefont {Vidick}}]{brakerski2021cryptographic}%
  \BibitemOpen
  \bibfield  {author} {\bibinfo {author} {\bibfnamefont {Z.}~\bibnamefont {Brakerski}}, \bibinfo {author} {\bibfnamefont {P.}~\bibnamefont {Christiano}}, \bibinfo {author} {\bibfnamefont {U.}~\bibnamefont {Mahadev}}, \bibinfo {author} {\bibfnamefont {U.}~\bibnamefont {Vazirani}},\ and\ \bibinfo {author} {\bibfnamefont {T.}~\bibnamefont {Vidick}},\ }\bibfield  {title} {\bibinfo {title} {A cryptographic test of quantumness and certifiable randomness from a single quantum device},\ }\href@noop {} {\bibfield  {journal} {\bibinfo  {journal} {Journal of the ACM (JACM)}\ }\textbf {\bibinfo {volume} {68}},\ \bibinfo {pages} {1} (\bibinfo {year} {2021})}\BibitemShut {NoStop}%
\bibitem [{\citenamefont {Buhrman}\ \emph {et~al.}(2001)\citenamefont {Buhrman}, \citenamefont {Cleve}, \citenamefont {Watrous},\ and\ \citenamefont {de~Wolf}}]{buhrman2001quantum}%
  \BibitemOpen
  \bibfield  {author} {\bibinfo {author} {\bibfnamefont {H.}~\bibnamefont {Buhrman}}, \bibinfo {author} {\bibfnamefont {R.}~\bibnamefont {Cleve}}, \bibinfo {author} {\bibfnamefont {J.}~\bibnamefont {Watrous}},\ and\ \bibinfo {author} {\bibfnamefont {R.}~\bibnamefont {de~Wolf}},\ }\bibfield  {title} {\bibinfo {title} {Quantum fingerprinting},\ }\href@noop {} {\bibfield  {journal} {\bibinfo  {journal} {Phys. Rev. Lett.}\ }\textbf {\bibinfo {volume} {87}},\ \bibinfo {pages} {167902} (\bibinfo {year} {2001})}\BibitemShut {NoStop}%
\bibitem [{\citenamefont {Gilboa}\ and\ \citenamefont {McClean}(2023)}]{gilboa2023exponential}%
  \BibitemOpen
  \bibfield  {author} {\bibinfo {author} {\bibfnamefont {D.}~\bibnamefont {Gilboa}}\ and\ \bibinfo {author} {\bibfnamefont {J.~R.}\ \bibnamefont {McClean}},\ }\bibfield  {title} {\bibinfo {title} {Exponential quantum communication advantage in distributed learning},\ }\href@noop {} {\bibfield  {journal} {\bibinfo  {journal} {arXiv preprint arXiv:2310.07136}\ } (\bibinfo {year} {2023})}\BibitemShut {NoStop}%
\bibitem [{\citenamefont {Bshouty}\ and\ \citenamefont {Jackson}(1995)}]{bshouty1995learning}%
  \BibitemOpen
  \bibfield  {author} {\bibinfo {author} {\bibfnamefont {N.~H.}\ \bibnamefont {Bshouty}}\ and\ \bibinfo {author} {\bibfnamefont {J.~C.}\ \bibnamefont {Jackson}},\ }\bibfield  {title} {\bibinfo {title} {Learning {DNF} over the uniform distribution using a quantum example oracle},\ }in\ \href@noop {} {\emph {\bibinfo {booktitle} {Proceedings of the eighth annual conference on Computational learning theory}}}\ (\bibinfo {year} {1995})\ pp.\ \bibinfo {pages} {118--127}\BibitemShut {NoStop}%
\bibitem [{\citenamefont {Grilo}\ \emph {et~al.}(2019)\citenamefont {Grilo}, \citenamefont {Kerenidis},\ and\ \citenamefont {Zijlstra}}]{grilo2019learning}%
  \BibitemOpen
  \bibfield  {author} {\bibinfo {author} {\bibfnamefont {A.~B.}\ \bibnamefont {Grilo}}, \bibinfo {author} {\bibfnamefont {I.}~\bibnamefont {Kerenidis}},\ and\ \bibinfo {author} {\bibfnamefont {T.}~\bibnamefont {Zijlstra}},\ }\bibfield  {title} {\bibinfo {title} {Learning-with-errors problem is easy with quantum samples},\ }\href@noop {} {\bibfield  {journal} {\bibinfo  {journal} {Phys. Rev. A}\ }\textbf {\bibinfo {volume} {99}},\ \bibinfo {pages} {032314} (\bibinfo {year} {2019})}\BibitemShut {NoStop}%
\bibitem [{\citenamefont {Gavinsky}\ \emph {et~al.}(2007)\citenamefont {Gavinsky}, \citenamefont {Kempe}, \citenamefont {Kerenidis}, \citenamefont {Raz},\ and\ \citenamefont {de~Wolf}}]{gavinsky2007exponential}%
  \BibitemOpen
  \bibfield  {author} {\bibinfo {author} {\bibfnamefont {D.}~\bibnamefont {Gavinsky}}, \bibinfo {author} {\bibfnamefont {J.}~\bibnamefont {Kempe}}, \bibinfo {author} {\bibfnamefont {I.}~\bibnamefont {Kerenidis}}, \bibinfo {author} {\bibfnamefont {R.}~\bibnamefont {Raz}},\ and\ \bibinfo {author} {\bibfnamefont {R.}~\bibnamefont {de~Wolf}},\ }\bibfield  {title} {\bibinfo {title} {Exponential separations for one-way quantum communication complexity, with applications to cryptography},\ }in\ \href@noop {} {\emph {\bibinfo {booktitle} {Proceedings of the thirty-ninth annual ACM symposium on Theory of computing}}}\ (\bibinfo {year} {2007})\ pp.\ \bibinfo {pages} {516--525}\BibitemShut {NoStop}%
\bibitem [{\citenamefont {Kallaugher}\ \emph {et~al.}(2024)\citenamefont {Kallaugher}, \citenamefont {Parekh},\ and\ \citenamefont {Voronova}}]{kallaugher2024exponential}%
  \BibitemOpen
  \bibfield  {author} {\bibinfo {author} {\bibfnamefont {J.}~\bibnamefont {Kallaugher}}, \bibinfo {author} {\bibfnamefont {O.}~\bibnamefont {Parekh}},\ and\ \bibinfo {author} {\bibfnamefont {N.}~\bibnamefont {Voronova}},\ }\bibfield  {title} {\bibinfo {title} {Exponential quantum space advantage for approximating maximum directed cut in the streaming model},\ }in\ \href@noop {} {\emph {\bibinfo {booktitle} {Proceedings of the 56th Annual ACM Symposium on Theory of Computing}}}\ (\bibinfo {year} {2024})\ pp.\ \bibinfo {pages} {1805--1815}\BibitemShut {NoStop}%
\bibitem [{\citenamefont {Haussler}(1992)}]{Haussler_1992}%
  \BibitemOpen
  \bibfield  {author} {\bibinfo {author} {\bibfnamefont {D.}~\bibnamefont {Haussler}},\ }\bibfield  {title} {\bibinfo {title} {Decision theoretic generalizations of the pac model for neural net and other learning applications},\ }\href@noop {} {\bibfield  {journal} {\bibinfo  {journal} {Information and Computation}\ }\textbf {\bibinfo {volume} {100}},\ \bibinfo {pages} {78} (\bibinfo {year} {1992})}\BibitemShut {NoStop}%
\bibitem [{\citenamefont {Hanneke}(2016)}]{Hanneke_2016}%
  \BibitemOpen
  \bibfield  {author} {\bibinfo {author} {\bibfnamefont {S.}~\bibnamefont {Hanneke}},\ }\bibfield  {title} {\bibinfo {title} {The optimal sample complexity of pac learning},\ }\href@noop {} {\bibfield  {journal} {\bibinfo  {journal} {Journal of Machine Learning Research}\ }\textbf {\bibinfo {volume} {17}},\ \bibinfo {pages} {1} (\bibinfo {year} {2016})}\BibitemShut {NoStop}%
\bibitem [{\citenamefont {Aaronson}(2007)}]{Aaronson_2007}%
  \BibitemOpen
  \bibfield  {author} {\bibinfo {author} {\bibfnamefont {S.}~\bibnamefont {Aaronson}},\ }\bibfield  {title} {\bibinfo {title} {The learnability of quantum states},\ }\href@noop {} {\bibfield  {journal} {\bibinfo  {journal} {Proceedings of the Royal Society A: Mathematical, Physical and Engineering Sciences}\ }\textbf {\bibinfo {volume} {463}},\ \bibinfo {pages} {3089} (\bibinfo {year} {2007})}\BibitemShut {NoStop}%
\bibitem [{\citenamefont {Arunachalam}\ and\ \citenamefont {de~Wolf}(2017)}]{arunachalam2017guest}%
  \BibitemOpen
  \bibfield  {author} {\bibinfo {author} {\bibfnamefont {S.}~\bibnamefont {Arunachalam}}\ and\ \bibinfo {author} {\bibfnamefont {R.}~\bibnamefont {de~Wolf}},\ }\bibfield  {title} {\bibinfo {title} {Guest column: A survey of quantum learning theory},\ }\href@noop {} {\bibfield  {journal} {\bibinfo  {journal} {ACM Sigact News}\ }\textbf {\bibinfo {volume} {48}},\ \bibinfo {pages} {41} (\bibinfo {year} {2017})}\BibitemShut {NoStop}%
\bibitem [{\citenamefont {Coopmans}\ and\ \citenamefont {Benedetti}(2024)}]{Coopmans_2024}%
  \BibitemOpen
  \bibfield  {author} {\bibinfo {author} {\bibfnamefont {L.}~\bibnamefont {Coopmans}}\ and\ \bibinfo {author} {\bibfnamefont {M.}~\bibnamefont {Benedetti}},\ }\bibfield  {title} {\bibinfo {title} {On the sample complexity of quantum boltzmann machine learning},\ }\href@noop {} {\bibfield  {journal} {\bibinfo  {journal} {Communications Physics}\ }\textbf {\bibinfo {volume} {7}} (\bibinfo {year} {2024})}\BibitemShut {NoStop}%
\bibitem [{\citenamefont {Arunachalam}\ \emph {et~al.}(2020)\citenamefont {Arunachalam}, \citenamefont {Belovs}, \citenamefont {Childs}, \citenamefont {Kothari}, \citenamefont {Rosmanis},\ and\ \citenamefont {de~Wolf}}]{Arunachalam_2020}%
  \BibitemOpen
  \bibfield  {author} {\bibinfo {author} {\bibfnamefont {S.}~\bibnamefont {Arunachalam}}, \bibinfo {author} {\bibfnamefont {A.}~\bibnamefont {Belovs}}, \bibinfo {author} {\bibfnamefont {A.~M.}\ \bibnamefont {Childs}}, \bibinfo {author} {\bibfnamefont {R.}~\bibnamefont {Kothari}}, \bibinfo {author} {\bibfnamefont {A.}~\bibnamefont {Rosmanis}},\ and\ \bibinfo {author} {\bibfnamefont {R.}~\bibnamefont {de~Wolf}},\ }\bibfield  {title} {\bibinfo {title} {{Quantum Coupon Collector}},\ }in\ \href@noop {} {\emph {\bibinfo {booktitle} {15th Conference on the Theory of Quantum Computation, Communication and Cryptography (TQC 2020)}}},\ Vol.\ \bibinfo {volume} {158}\ (\bibinfo {year} {2020})\ pp.\ \bibinfo {pages} {10:1--10:17}\BibitemShut {NoStop}%
\bibitem [{\citenamefont {Holevo}(1973)}]{holevo1973bounds}%
  \BibitemOpen
  \bibfield  {author} {\bibinfo {author} {\bibfnamefont {A.~S.}\ \bibnamefont {Holevo}},\ }\bibfield  {title} {\bibinfo {title} {Bounds for the quantity of information transmitted by a quantum communication channel},\ }\href@noop {} {\bibfield  {journal} {\bibinfo  {journal} {Problemy Peredachi Informatsii}\ }\textbf {\bibinfo {volume} {9}},\ \bibinfo {pages} {3} (\bibinfo {year} {1973})}\BibitemShut {NoStop}%
\bibitem [{\citenamefont {Aaronson}\ and\ \citenamefont {Ambainis}(2015)}]{aaronson2015forrelation}%
  \BibitemOpen
  \bibfield  {author} {\bibinfo {author} {\bibfnamefont {S.}~\bibnamefont {Aaronson}}\ and\ \bibinfo {author} {\bibfnamefont {A.}~\bibnamefont {Ambainis}},\ }\bibfield  {title} {\bibinfo {title} {Forrelation: A problem that optimally separates quantum from classical computing},\ }in\ \href@noop {} {\emph {\bibinfo {booktitle} {Proceedings of the Forty-Seventh Annual ACM Symposium on Theory of Computing}}},\ \bibinfo {series and number} {STOC '15}\ (\bibinfo  {publisher} {Association for Computing Machinery},\ \bibinfo {address} {New York, NY, USA},\ \bibinfo {year} {2015})\ p.\ \bibinfo {pages} {307–316}\BibitemShut {NoStop}%
\bibitem [{Note1()}]{Note1}%
  \BibitemOpen
  \bibinfo {note} {They proved that the separation can be at most $1$ versus $\protect \tilde {\Omega }(\protect \sqrt {N})$, which is achieved by the Forrelation problem. For a constant number of queries, one can achieve $\protect \mathcal {O}_\epsilon (1)$ versus $\Omega (N^{1-\epsilon })$ for any $\epsilon >0$~\cite {bansal2021k}. Here $\protect \mathcal {O}_{\epsilon }(f(n))$ means $\protect \mathcal {O}(f(n))$ where the hidden constant is allowed to depend on $\epsilon $, which is treated as a fixed parameter.}\BibitemShut {Stop}%
\bibitem [{\citenamefont {Aaronson}\ \emph {et~al.}(2020)\citenamefont {Aaronson}, \citenamefont {Atia},\ and\ \citenamefont {Susskind}}]{aaronson2020hardness}%
  \BibitemOpen
  \bibfield  {author} {\bibinfo {author} {\bibfnamefont {S.}~\bibnamefont {Aaronson}}, \bibinfo {author} {\bibfnamefont {Y.}~\bibnamefont {Atia}},\ and\ \bibinfo {author} {\bibfnamefont {L.}~\bibnamefont {Susskind}},\ }\bibfield  {title} {\bibinfo {title} {On the hardness of detecting macroscopic superpositions},\ }\href@noop {} {\bibfield  {journal} {\bibinfo  {journal} {arXiv preprint arXiv:2009.07450}\ } (\bibinfo {year} {2020})}\BibitemShut {NoStop}%
\bibitem [{\citenamefont {Deutsch}\ and\ \citenamefont {Jozsa}(1992)}]{deutsch1992rapid}%
  \BibitemOpen
  \bibfield  {author} {\bibinfo {author} {\bibfnamefont {D.}~\bibnamefont {Deutsch}}\ and\ \bibinfo {author} {\bibfnamefont {R.}~\bibnamefont {Jozsa}},\ }\bibfield  {title} {\bibinfo {title} {Rapid solution of problems by quantum computation},\ }\href@noop {} {\bibfield  {journal} {\bibinfo  {journal} {Proceedings of the Royal Society of London. Series A: Mathematical and Physical Sciences}\ }\textbf {\bibinfo {volume} {439}},\ \bibinfo {pages} {553} (\bibinfo {year} {1992})}\BibitemShut {NoStop}%
\bibitem [{\citenamefont {Helstrom}(1969)}]{helstrom1969quantum}%
  \BibitemOpen
  \bibfield  {author} {\bibinfo {author} {\bibfnamefont {C.~W.}\ \bibnamefont {Helstrom}},\ }\bibfield  {title} {\bibinfo {title} {Quantum detection and estimation theory},\ }\href@noop {} {\bibfield  {journal} {\bibinfo  {journal} {Journal of Statistical Physics}\ }\textbf {\bibinfo {volume} {1}},\ \bibinfo {pages} {231} (\bibinfo {year} {1969})}\BibitemShut {NoStop}%
\bibitem [{Note2()}]{Note2}%
  \BibitemOpen
  \bibinfo {note} {This circuit lower bound was to be expected under the assumption that $\protect \mathsf {QMA}\protect \neq \protect \mathsf {QCMA}$, because if all subset states could be (approximately) constructed with a polynomially-sized circuit then we would have that $\protect \mathsf {QCMA}= \protect \mathsf {QMA}$ by the result of~\cite {grilo2015qma}.}\BibitemShut {Stop}%
\bibitem [{\citenamefont {Nie}\ \emph {et~al.}(2024)\citenamefont {Nie}, \citenamefont {Zi},\ and\ \citenamefont {Sun}}]{nie2024quantum}%
  \BibitemOpen
  \bibfield  {author} {\bibinfo {author} {\bibfnamefont {J.}~\bibnamefont {Nie}}, \bibinfo {author} {\bibfnamefont {W.}~\bibnamefont {Zi}},\ and\ \bibinfo {author} {\bibfnamefont {X.}~\bibnamefont {Sun}},\ }\bibfield  {title} {\bibinfo {title} {Quantum circuit for multi-qubit {Toffoli} gate with optimal resource},\ }\href@noop {} {\bibfield  {journal} {\bibinfo  {journal} {arXiv preprint arXiv:2402.05053}\ } (\bibinfo {year} {2024})}\BibitemShut {NoStop}%
\bibitem [{\citenamefont {Wang}\ \emph {et~al.}(2022)\citenamefont {Wang}, \citenamefont {Wei},\ and\ \citenamefont {Long}}]{wang2022quantum}%
  \BibitemOpen
  \bibfield  {author} {\bibinfo {author} {\bibfnamefont {Z.-G.}\ \bibnamefont {Wang}}, \bibinfo {author} {\bibfnamefont {S.-J.}\ \bibnamefont {Wei}},\ and\ \bibinfo {author} {\bibfnamefont {G.-L.}\ \bibnamefont {Long}},\ }\bibfield  {title} {\bibinfo {title} {A quantum circuit design of {AES} requiring fewer quantum qubits and gate operations},\ }\href@noop {} {\bibfield  {journal} {\bibinfo  {journal} {Frontiers of Physics}\ }\textbf {\bibinfo {volume} {17}},\ \bibinfo {pages} {41501} (\bibinfo {year} {2022})}\BibitemShut {NoStop}%
\bibitem [{\citenamefont {Gay}\ \emph {et~al.}(2025)\citenamefont {Gay}, \citenamefont {He}, \citenamefont {Kocurek},\ and\ \citenamefont {O'Donnell}}]{gay2025pseudorandomness}%
  \BibitemOpen
  \bibfield  {author} {\bibinfo {author} {\bibfnamefont {W.}~\bibnamefont {Gay}}, \bibinfo {author} {\bibfnamefont {W.}~\bibnamefont {He}}, \bibinfo {author} {\bibfnamefont {N.}~\bibnamefont {Kocurek}},\ and\ \bibinfo {author} {\bibfnamefont {R.}~\bibnamefont {O'Donnell}},\ }\bibfield  {title} {\bibinfo {title} {Pseudorandomness properties of random reversible circuits},\ }in\ \href@noop {} {\emph {\bibinfo {booktitle} {Advances in Cryptology -- CRYPTO 2025}}}\ (\bibinfo  {publisher} {Springer Nature Switzerland},\ \bibinfo {address} {Cham},\ \bibinfo {year} {2025})\ pp.\ \bibinfo {pages} {651--678}\BibitemShut {NoStop}%
\bibitem [{\citenamefont {Knill}(2004)}]{knill2004postselected}%
  \BibitemOpen
  \bibfield  {author} {\bibinfo {author} {\bibfnamefont {E.}~\bibnamefont {Knill}},\ }\bibfield  {title} {\bibinfo {title} {Fault-tolerant postselected quantum computation: Threshold analysis},\ }\href@noop {} {\bibfield  {journal} {\bibinfo  {journal} {arXiv preprint quant-ph/0404104}\ } (\bibinfo {year} {2004})}\BibitemShut {NoStop}%
\bibitem [{\citenamefont {Self}\ \emph {et~al.}(2024)\citenamefont {Self}, \citenamefont {Benedetti},\ and\ \citenamefont {Amaro}}]{self2024protecting}%
  \BibitemOpen
  \bibfield  {author} {\bibinfo {author} {\bibfnamefont {C.~N.}\ \bibnamefont {Self}}, \bibinfo {author} {\bibfnamefont {M.}~\bibnamefont {Benedetti}},\ and\ \bibinfo {author} {\bibfnamefont {D.}~\bibnamefont {Amaro}},\ }\bibfield  {title} {\bibinfo {title} {Protecting expressive circuits with a quantum error detection code},\ }\href@noop {} {\bibfield  {journal} {\bibinfo  {journal} {Nature Physics}\ }\textbf {\bibinfo {volume} {20}},\ \bibinfo {pages} {219–224} (\bibinfo {year} {2024})}\BibitemShut {NoStop}%
\bibitem [{\citenamefont {Benedetti}\ \emph {et~al.}(2025)\citenamefont {Benedetti}, \citenamefont {Marin-Sanchez}, \citenamefont {Weggemans}, \citenamefont {Rosenkranz},\ and\ \citenamefont {Buhrman}}]{benedetti2025unconditional}%
  \BibitemOpen
  \bibfield  {author} {\bibinfo {author} {\bibfnamefont {M.}~\bibnamefont {Benedetti}}, \bibinfo {author} {\bibfnamefont {G.}~\bibnamefont {Marin-Sanchez}}, \bibinfo {author} {\bibfnamefont {J.}~\bibnamefont {Weggemans}}, \bibinfo {author} {\bibfnamefont {M.}~\bibnamefont {Rosenkranz}},\ and\ \bibinfo {author} {\bibfnamefont {H.}~\bibnamefont {Buhrman}},\ }\bibfield  {title} {\bibinfo {title} {Unconditional and exponentially large violation of classicality},\ }\href@noop {} {\bibfield  {journal} {\bibinfo  {journal} {arXiv preprint arXiv:2511.11008}\ } (\bibinfo {year} {2025})}\BibitemShut {NoStop}%
\bibitem [{\citenamefont {Bansal}\ and\ \citenamefont {Sinha}(2021)}]{bansal2021k}%
  \BibitemOpen
  \bibfield  {author} {\bibinfo {author} {\bibfnamefont {N.}~\bibnamefont {Bansal}}\ and\ \bibinfo {author} {\bibfnamefont {M.}~\bibnamefont {Sinha}},\ }\bibfield  {title} {\bibinfo {title} {K-forrelation optimally separates quantum and classical query complexity},\ }in\ \href@noop {} {\emph {\bibinfo {booktitle} {Proceedings of the 53rd Annual ACM SIGACT Symposium on Theory of Computing}}}\ (\bibinfo {year} {2021})\ pp.\ \bibinfo {pages} {1303--1316}\BibitemShut {NoStop}%
\bibitem [{\citenamefont {Grilo}\ \emph {et~al.}(2015)\citenamefont {Grilo}, \citenamefont {Kerenidis},\ and\ \citenamefont {Sikora}}]{grilo2015qma}%
  \BibitemOpen
  \bibfield  {author} {\bibinfo {author} {\bibfnamefont {A.~B.}\ \bibnamefont {Grilo}}, \bibinfo {author} {\bibfnamefont {I.}~\bibnamefont {Kerenidis}},\ and\ \bibinfo {author} {\bibfnamefont {J.}~\bibnamefont {Sikora}},\ }\bibfield  {title} {\bibinfo {title} {{QMA} with subset state witnesses},\ }in\ \href@noop {} {\emph {\bibinfo {booktitle} {International Symposium on Mathematical Foundations of Computer Science}}}\ (\bibinfo {organization} {Springer},\ \bibinfo {year} {2015})\ pp.\ \bibinfo {pages} {163--174}\BibitemShut {NoStop}%
\bibitem [{\citenamefont {Grover}(1996)}]{grover1996fast}%
  \BibitemOpen
  \bibfield  {author} {\bibinfo {author} {\bibfnamefont {L.~K.}\ \bibnamefont {Grover}},\ }\bibfield  {title} {\bibinfo {title} {A fast quantum mechanical algorithm for database search},\ }in\ \href@noop {} {\emph {\bibinfo {booktitle} {Proceedings of the twenty-eighth annual ACM symposium on Theory of computing}}}\ (\bibinfo {year} {1996})\ pp.\ \bibinfo {pages} {212--219}\BibitemShut {NoStop}%
\bibitem [{\citenamefont {Zalka}(1999)}]{zalka1999grover}%
  \BibitemOpen
  \bibfield  {author} {\bibinfo {author} {\bibfnamefont {C.}~\bibnamefont {Zalka}},\ }\bibfield  {title} {\bibinfo {title} {A {Grover}-based quantum search of optimal order for an unknown number of marked elements},\ }\href@noop {} {\bibfield  {journal} {\bibinfo  {journal} {arXiv preprint quant-ph/9902049}\ } (\bibinfo {year} {1999})}\BibitemShut {NoStop}%
\bibitem [{\citenamefont {Yao}(1977)}]{yao1977probabilistic}%
  \BibitemOpen
  \bibfield  {author} {\bibinfo {author} {\bibfnamefont {A.~C.-C.}\ \bibnamefont {Yao}},\ }\bibfield  {title} {\bibinfo {title} {Probabilistic computations: Toward a unified measure of complexity},\ }in\ \href@noop {} {\emph {\bibinfo {booktitle} {18th Annual Symposium on Foundations of Computer Science (sfcs 1977)}}}\ (\bibinfo {organization} {IEEE Computer Society},\ \bibinfo {year} {1977})\ pp.\ \bibinfo {pages} {222--227}\BibitemShut {NoStop}%
\bibitem [{\citenamefont {Mendelson}\ \emph {et~al.}(2016)\citenamefont {Mendelson}, \citenamefont {Zuluaga}, \citenamefont {Hutton},\ and\ \citenamefont {Ourselin}}]{mendelson2016distribution}%
  \BibitemOpen
  \bibfield  {author} {\bibinfo {author} {\bibfnamefont {A.~F.}\ \bibnamefont {Mendelson}}, \bibinfo {author} {\bibfnamefont {M.~A.}\ \bibnamefont {Zuluaga}}, \bibinfo {author} {\bibfnamefont {B.~F.}\ \bibnamefont {Hutton}},\ and\ \bibinfo {author} {\bibfnamefont {S.}~\bibnamefont {Ourselin}},\ }\bibfield  {title} {\bibinfo {title} {What is the distribution of the number of unique original items in a bootstrap sample?},\ }\href@noop {} {\bibfield  {journal} {\bibinfo  {journal} {arXiv preprint arXiv:1602.05822}\ } (\bibinfo {year} {2016})}\BibitemShut {NoStop}%
\bibitem [{\citenamefont {Luby}\ and\ \citenamefont {Rackoff}(1988)}]{luby1988how}%
  \BibitemOpen
  \bibfield  {author} {\bibinfo {author} {\bibfnamefont {M.}~\bibnamefont {Luby}}\ and\ \bibinfo {author} {\bibfnamefont {C.}~\bibnamefont {Rackoff}},\ }\bibfield  {title} {\bibinfo {title} {How to construct pseudorandom permutations from pseudorandom functions},\ }\href@noop {} {\bibfield  {journal} {\bibinfo  {journal} {SIAM Journal on Computing}\ }\textbf {\bibinfo {volume} {17}},\ \bibinfo {pages} {373} (\bibinfo {year} {1988})}\BibitemShut {NoStop}%
\bibitem [{\citenamefont {Katz}\ and\ \citenamefont {Lindell}(2020)}]{katz2020introduction}%
  \BibitemOpen
  \bibfield  {author} {\bibinfo {author} {\bibfnamefont {J.}~\bibnamefont {Katz}}\ and\ \bibinfo {author} {\bibfnamefont {Y.}~\bibnamefont {Lindell}},\ }\href@noop {} {\emph {\bibinfo {title} {Introduction to Modern Cryptography}}},\ Chapman \& Hall/CRC Cryptography and Network Security Series\ (\bibinfo  {publisher} {CRC Press},\ \bibinfo {year} {2020})\BibitemShut {NoStop}%
\bibitem [{\citenamefont {Ray-Chaudhuri}\ and\ \citenamefont {Wilson}(1975)}]{ray1975t}%
  \BibitemOpen
  \bibfield  {author} {\bibinfo {author} {\bibfnamefont {D.~K.}\ \bibnamefont {Ray-Chaudhuri}}\ and\ \bibinfo {author} {\bibfnamefont {R.~M.}\ \bibnamefont {Wilson}},\ }\bibfield  {title} {\bibinfo {title} {{On $t$-designs}},\ }\href@noop {} {\bibfield  {journal} {\bibinfo  {journal} {Osaka Journal of Mathematics}\ }\textbf {\bibinfo {volume} {12}},\ \bibinfo {pages} {737 } (\bibinfo {year} {1975})}\BibitemShut {NoStop}%
\end{thebibliography}%
\let\addcontentsline\oldaddcontentsline

\clearpage
\onecolumngrid
\begin{center}
	\textbf{\large{Supplemental Material for}}\\
    \textbf{\large{``Provable and Verifiable Quantum Advantage in Sample Complexity''}}\\
    \vspace{.5cm}
    Marcello Benedetti,\textsuperscript{1} Harry Buhrman,\textsuperscript{1, 2, 3} and Jordi Weggemans\textsuperscript{2, 4}\\
    \small{\textit{$^1$Quantinuum, London, United Kingdom\\
    $^2$QuSoft, Amsterdam, The Netherlands\\
    $^3$University of Amsterdam, Amsterdam, The Netherlands\\
    $^4$CWI, Amsterdam, The Netherlands}}
\end{center}

\tableofcontents

\vspace{1cm}

\section{Complement sampling using quantum samples}
\label{sec:cs_quantum}

We consider the following sampling problem. For some $n \in \mathbb{N}$, let $S \subset \omega:= \{0,1\}^n $ be a non-empty subset of cardinality $K$ out of a set of $N := 2^n$ elements. Let $\bar{S} = \omega \setminus S $ be the complement set with cardinality $N-K$. Given access to a quantum sample $\ket{S}$ of the uniform distribution of the elements in $S$, the task of complement sampling is to return \emph{any} element from $\bar{S}$. 

In this section we discuss solutions to this problem where a quantum computer swaps the input state as 
\begin{equation*}
    \ket{S} = \frac{1}{\sqrt{K}} \sum_{x \in S} \ket{x} \qquad \longrightarrow \qquad \ket{\bar{S}} = \frac{1}{\sqrt{N - K}} \sum_{x \notin S} \ket{x} .
\end{equation*}
Then, a measurement in the computational basis yields a uniformly random sample from $\bar{S}$.

Hereafter, we use Landau symbols $\mathcal{O}$ and $\Omega$ for asymptotic upper and lower bounds, respectively. For cases where polylogarithmic terms are omitted, we use $\tilde{\mathcal{O}}$ and $\tilde{\Omega}$ instead. For $N$ fixed, we write $\mathcal{S}_K$ for the family of all subsets $S \subset \omega$ with cardinality $K$.

\subsection{Aaronson, Atia and Susskind's construction}
\label{ss:AAS}

We will start by reviewing Aaronson, Atia and Susskind's construction~\cite{aaronson2020hardness}, which formed our basis to arrive at our solution for the quantum sample variant to complement sampling. For most results in this work, except for~\cref{sec:quantum_cc_and_d}, we found alternative proofs and easier constructions without relying on this construction. However, we think that poofs based on this construction provide the best intuition for why complement sampling can be perfectly solved on a quantum computer when $K=N/2$.

We start by recalling the four notions of complexities on quantum states from~\cite{aaronson2020hardness}: relative complexity, circuit complexity, swap complexity and distinguishability complexity.

\begin{definition}[Relative complexity]
    The relative complexity $\mathcal{C}_\epsilon (\ket{a},\ket{b})$ of two $n$-qubit pure quantum states $\ket{a}$,$\ket{b}$ is defined as the minimal number of gates in a circuit $C$ such that 
    \begin{align*}
        \abs{ \bra{b} \bra{0 \dots 0} C \ket{a} \ket{0 \dots 0}} \geq 1-\epsilon.
    \end{align*}
\end{definition}

\begin{definition}[Circuit complexity]
    The circuit complexity $\mathcal{C}_\epsilon (\ket{a})$ of a $n$-qubit pure quantum state $\ket{a}$ is defined as the relative complexity with the $n$-qubit all zero state $\ket{0 \dots 0}$, i.e.,
    \begin{align*}
        \mathcal{C}_\epsilon (\ket{a}) =  \mathcal{C}_\epsilon (\ket{0 \dots 0},\ket{a}).
    \end{align*}
\end{definition}

\begin{definition}[Swap complexity]
    The swap complexity $\mathcal{S}_\epsilon(\ket{a},\ket{b})$ of two $n$-qubit pure quantum states $\ket{a}$,$\ket{b}$ is defined as the minimal number of gates in a circuit $C$ such that 
    \begin{align*}
        \frac{1}{2} \big| \bra{a} \bra{0 \dots 0} C \ket{b} \ket{0 \dots 0} + \bra{b} \bra{0 \dots 0} C \ket{a} \ket{0 \dots 0} \big| \geq 1-\epsilon .
    \end{align*}
    \label{def:swap_complexity}
\end{definition}

\noindent It can be verified that $\mathcal{S}_\epsilon \geq \mathcal{C}_\epsilon$ holds for any $0 \leq \epsilon \leq 1$~\cite{aaronson2020hardness}.

\begin{definition}[Distinguishability complexity]
    The distinguishability complexity $\mathcal{D}_\epsilon(\ket{a},\ket{b})$ of two $n$-qubit pure quantum states $\ket{a}$,$\ket{b}$ is defined as the minimal number of gates in a circuit $C$ such that 
    \begin{align*}
        \abs{\norm{ \Pi_1 C \ket{a} \ket{0 \dots 0} }^2_2 -\norm{ \Pi_1 C \ket{b} \ket{0 \dots 0}}^2_2} \geq 1- \epsilon,
    \end{align*}
    where $\Pi_1 = \dyad{1} \otimes \mathbb{I}$.
    \label{def:dist_compl}
\end{definition}
The key result from~\cite{aaronson2020hardness} is an equivalence (in terms of the order of circuit complexity) between a perfect swapper of two orthogonal states, and a perfect distinguisher for the corresponding conjugate states. 

\begin{lemma}[Adapted from~\cite{aaronson2020hardness}, Theorem 2] 
    Let $0 \leq \epsilon < 1$ and $n \in \mathbb{N}$. Let $\ket{a}$, $\ket{b}$ be orthogonal $n$-qubit states and let $\ket{\phi^+} = \frac{1}{\sqrt{2}}(\ket{a}+\ket{b})$ and $\ket{\phi^-} = \frac{1}{\sqrt{2}}(\ket{a}-\ket{b})$. The following holds:
    \item If $\mathcal{S}_\epsilon (\ket{a},\ket{b}) = T_1$, then we have that $\mathcal{D}_\epsilon (\ket{\phi^+},\ket{\phi^-}) = \mathcal{O}(T_1)$.
    \item If $\mathcal{D}_\epsilon (\ket{\phi^+},\ket{\phi^-}) = T_2$, then we have that $ \mathcal{S}_\epsilon (\ket{a},\ket{b})= \mathcal{O}(T_2)$.
    \label{lem:S_versus_D}
\end{lemma}

\begin{figure}
    \centering
    \begin{tabular}{ccc}
    (a) Distinguisher & \qquad & (b) Swapper \\
    \\
    \begin{quantikz}[wire types={q,q,q,n,q}, column sep=10pt, row sep={20pt,between origins}]
        \lstick[1]{$\ket{0}$} & \gate{H} & \ctrl{1} & \gate{H} & \meter \qw \rstick[1]{$+/-$} \\
        \lstick[4, brackets=none]{$\ket{\phi^+}/\ket{\phi^-}$} & \qw & \gate[wires=4][20pt]{U} & \qw & \qw \\
         & \qw & & \qw & \qw \\
         & \vdots & & \vdots &  \\
         & \qw & & \qw & \qw 
    \end{quantikz}
    &
    \qquad
    &
    \begin{quantikz}[wire types={q,q,n,q},column sep=10pt, row sep={20pt,between origins}]
        \lstick[4, brackets=none]{$\ket{a}/\ket{b}$} & \qw & \gate[wires=4][20pt]{A} & \gate{Z} & \gate[wires=4][20pt]{A^\dag} & \qw & \rstick[4, brackets=none]{$\ket{b}/\ket{a}$} \\
         & & & & & & \\
         & \vdots & & \vdots & & \vdots & \\
         & & & & & &
    \end{quantikz}
    \end{tabular}
    \caption{Quantum circuits for Aaronson, Atia and Susskind’s construction~\cite{aaronson2020hardness}. (a) Circuit distinguishing $\ket{\phi^+}$ from $\ket{\phi^-}$ using a unitary $U$ that swaps $\ket{a}$ and $\ket{b}$. (b) Circuit swapping $\ket{a}$ and $\ket{b}$ using a unitary $A$ that distinguishes $\ket{\phi^+}$ and $\ket{\phi^-}$.}
    \label{fig:aas_construction}
\end{figure}

\noindent An intuition for this Lemma is given in~\cref{fig:aas_construction} where the distinguisher and swapper circuits are defined in terms of each other. This indicates that the number of gates in these circuits should be of the same order.

The following claim is also given in~\cite[Corollary 1]{aaronson2020hardness}, but we include it to specify the values of $\epsilon$, as this will be needed in~\cref{sec:quantum_cc_and_d}.

\begin{proposition}
    Let $\mathcal{G}$ be a finite, self-inverse gate set. Let $\ket{a},\ket{b}$ be two orthogonal $n$-qubit quantum states. If $
    \mathcal{D}_{4\epsilon} (\ket{a},\ket{b}) = T$, then $
    \mathcal{C}_{\epsilon} (\ket{a}) \geq T $ and $ \mathcal{C}_{\epsilon} (\ket{b}) \geq T$.
\label{prop:circuit_complexity_from_D}
\end{proposition}

\begin{proof}
We prove this by contradiction. Suppose, without loss of generality, that $ \mathcal{C}_\epsilon (\ket{a}) < T$. Then, there exists a quantum circuit $C$ with gate complexity strictly less than $T$ such that  
\begin{align*}
    \left| \bra{a} C \ket{0 \dots 0} \right| \geq 1 - \epsilon.
\end{align*}
Defining $\ket{o} = C \ket{0 \dots 0}$, we can write $\ket{o}$ as  
\begin{align*}
    \ket{o} = \sqrt{1-\alpha} e^{i \theta_1} \ket{a} + \sqrt{\alpha} e^{i \theta_2} \ket{a^\perp}
\end{align*}
for some $\theta_1, \theta_2 \in [0,2\pi]$ and $\alpha \in [0,1]$. The assumption $\left| \bra{a} \ket{o} \right| \geq 1-\epsilon$ implies that $\sqrt{1-\alpha} \geq 1-\epsilon$.
Now, consider the circuit that receives an unknown state $\ket{z}$, where $ z \in \{a,b\} $, applies $C^\dagger$, and measures in the computational basis. The probability of obtaining the all-zero outcome is then given by $\left| \bra{z} \ket{o} \right|^2$.
For $\ket{z} = \ket{a}$, we have  
\begin{align*}
    \left| \bra{a} \ket{o} \right|^2 \geq (1-\epsilon)^2.
\end{align*}
For $\ket{z} = \ket{b}$, orthogonality implies that  
\begin{align*}
    \left| \bra{b} \ket{o} \right|^2 = 1 - \left| \bra{a} \ket{o} \right|^2 \leq 1 - (1-\epsilon)^2.
\end{align*}
Thus, the bias of this measurement-based distinguisher is  
\begin{align*}
    \left| \left| \bra{a} \ket{o} \right|^2 - \left| \bra{b} \ket{o} \right|^2 \right|  
    \geq (1-\epsilon)^2 - (1 - (1-\epsilon)^2) = 1 - 4\epsilon + 2\epsilon^2 \geq 1 - 4\epsilon.
\end{align*}
Since $C^\dagger$ has the same gate complexity as $C$, this gives a distinguisher with gate complexity $< T$ and bias at least $1 - 4\epsilon$, contradicting the assumption that $\mathcal{D}_{4\epsilon} (\ket{a},\ket{b}) = T$. The same argument applies if $\mathcal{C}_\epsilon (\ket{b}) < T$ or both $\mathcal{C}_\epsilon (\ket{a})$ and $\mathcal{C}_\epsilon (\ket{b})$ are smaller than $T$. Thus, we conclude that $
\mathcal{C}_\epsilon (\ket{a}) \geq T $ and  $\mathcal{C}_\epsilon (\ket{b}) \geq T$ must hold.
\end{proof}

\subsection{A perfect swapper for subsets of cardinality \texorpdfstring{$K=N/2$}{K=N/2}}

Let us now show that the task of complement sampling, as stated in the beginning of~\cref{sec:cs_quantum}, can be perfectly solved using only a single quantum sample in the case where $K=N/2$.

\begin{theorem}[Quantum complement swapper] 
    Consider a subset $S \in \mathcal{S}_K$ with $K=N/2$. Then there exists a polynomial-time quantum algorithm which prepares the state $\ket{\bar{S}}$ from $\ket{S}$ and vice versa.
\label{thm:perfect_quantum_cs}
\end{theorem}

\begin{proof}
Any function $f : \omega \rightarrow \{0,1\}$ induces a quantum phase state of the form 
\begin{align*}
    \ket{y_{f}} = \frac{1}{\sqrt{N}} \sum_{x \in \omega} (-1)^{f(x)} \ket{x} .
\end{align*}
In particular, the constant function $f_\mathrm{con}(x) = 0$ yields the state $\ket{y_{f_\mathrm{con}}} = \frac{1}{\sqrt{N}} \sum_{x \in \{0,1\}^n} \ket{x} .$
For every $S$ and $\bar{S}$ with $K = N/2$ there exists a balanced function $f_\mathrm{bal}$ such that $f_\mathrm{bal}(x) = 0$ if $x \in S$ and $f_\mathrm{bal}(x) = 1$ otherwise. Then we can write 
\begin{align*}
    \ket{S} = \frac{1}{\sqrt{2}} \left( \ket{y_{f_\mathrm{con}}} + \ket{y_{f_\mathrm{bal}}} \right) , \qquad 
    \ket{\bar{S}} = \frac{1}{\sqrt{2}} \left( \ket{y_{f_\mathrm{con}}} - \ket{y_{f_\mathrm{bal}}} \right) .
\end{align*}
It is well known that the Deutsch–Jozsa algorithm~\cite{deutsch1992rapid} can perfectly distinguish the two phase states induced by constant and balanced functions. The Deutsch–Jozsa algorithm consists of a short quantum circuit $V$ that applies $H^{\otimes n}$ followed by the two-outcome measurement operator $\Lambda = \{\dyad{0 \dots 0}, I -\dyad{0 \dots 0} \}$. We copy the measurement outcome to an ancilla qubit using a multi-controlled NOT gate, where we control on the main register being in the zero state. This circuit is illustrated in~\cref{fig:dist_and_swap}~(a). Then, by~\cref{lem:S_versus_D}, there exists another quantum circuit $V'$ which only uses two applications of $V$ to swap $\ket{S}$ and $\ket{\bar{S}}$. This circuit is illustrated in~\cref{fig:dist_and_swap}~(b) where we ignore the ancilla qubit since, after circuit simplifications, it can be shown to idle in the $\ket{0}$ state throughout. Since $V'$ consists of two applications of $V$, the multi-controlled NOT gate and two $Z$ gates, it can be executed in time polynomial in $n$.
\end{proof}

Let us take a closer look at what the circuit used in the proof of~\cref{thm:perfect_quantum_cs} actually does, fixing the first qubit (so in the register which contains the qubit to which the outcome of $\Lambda$ is `copied') to be in $\ket{0}$ initially. Let the second $n$-qubit register be in an arbitrary state $\ket{\psi}$. A quick derivation shows that when the first qubit is fixed to be in $\ket{0}$ the circuit implements the mapping 
\begin{align*}
    \ket{0} \ket{\psi} \mapsto -\ket{0} U \ket{\psi},
\end{align*}
where $U = 2\dyad{+^n} - \mathbb{I}$ is the Grover diffusion operator~\cite{grover1996fast}. Noting that $\ket{+^n} = \frac{1}{\sqrt{N}} \left(\sqrt{N - K} \ket{\bar{S}} + \sqrt{K} \ket{S} \right)$, it can be verified that
\begin{align*}
    U \ket{S} &= 2\sqrt{\frac{K}{N}\left(1-\frac{K}{N}\right)} \ket{\bar{S}} + \left(2\frac{K}{N} - 1\right) \ket{S}.
\end{align*}
Thus, the Grover diffusion operator is an imperfect swapper $U \ket{S} = \sqrt{1- \epsilon} \ket{\bar{S}} + \sqrt{\epsilon} \ket{S}$ with error $\epsilon = \left(2\frac{K}{N} - 1\right)^2$. We see that zero error is achieved when $S$ contains exactly half of the elements, $K = N/2$.

\begin{figure}
\centering
\begin{tabular}{ccc}
(a) Distinguisher & \qquad & (b) Swapper \\
\\
\begin{quantikz}[wire types={q,q,q,n,q}, column sep=10pt, row sep={20pt,between origins}]
    \lstick[1]{$\ket{0}$} &  & \targ{}  & \\
    \lstick[4, brackets=none]{$\ket{y_{f_\mathrm{con}}} / \ket{y_{f_\mathrm{bal}}}$} & \gate{H} & \octrl{-1}  & \\
     & \gate{H} & \octrl{1}\wire[u]{q} &\\
     &          & \vdots               &\\
     & \gate{H} & \octrl{0}            &
\end{quantikz}
&
\qquad \qquad
&
\begin{quantikz}[wire types={q,q,n,q},column sep=10pt, row sep={20pt,between origins}]
    \lstick[4, brackets=none]{$\ket{S}/\ket{\bar{S}}$} & \gate{Z} & \targ{} & \gate{Z} & \rstick[4, brackets=none]{$\ket{\bar{S}}/\ket{S}$} \\
      &  \gate{H} & \octrl{1}\wire[u]{q} & \gate{H} & \qw \\
      &           & \vdots               &          &     \\
      &  \gate{H} & \octrl{0}            & \gate{H} & \qw
\end{quantikz}
\end{tabular}
\caption{Distinguisher and swapper for complement sampling. (a) Our distinguisher circuit can be thought of as the Deutsch-Jozsa algorithm followed by a copy of the measurement outcome into an ancilla qubit. (b) The swapper circuit is obtained by plugging our distinguisher circuit into~\cref{fig:aas_construction}~(b) and simplifying gates. The simplification leads to the ancilla qubit being idle, so we do not include it in this circuit diagram. Our swapper circuit corresponds to the Grover diffusion operator up to a global phase of $-1$. }
\label{fig:dist_and_swap}
\end{figure}

\subsection{Impossibility of a perfect swapper for \texorpdfstring{$K \neq N/2$}{K≠N/2}}
\label{sec:impossibility}

Consider again some fixed $N=2^n$ with $n \in  \mathbb{N}$. Let $K$ be the cardinality, and let $\mathcal{S}_K = \{ S : |S| = K\}$ be the family of all subsets $S \subset \omega$ that have cardinality $K$. We now prove that there exists no perfect swapper for any other case then $K = N/2$, given the condition that all auxiliary qubits must return to their initial state. In~\cref{sec:zero_error} we will give a simpler proof. However, the proof given here uses the same construction used in the previous subsections, providing more intuition into the problem at hand.

\begin{proposition} 
    Let $S \in \mathcal{S}_K$. Then for any $K \neq N/2$, there does not exist a quantum circuit which perfectly maps $\ket{S}$ to $\ket{\bar{S}}$ under the condition that all auxiliary qubits return to the $\ket{0}$ state.
\end{proposition}

\begin{proof}
We will prove our claim using a reductio ad absurdum: suppose that there exists a circuit $C$ which only depends on the cardinality of $S$ (and of course $N$), and can swap any $\ket{S}$ to $\ket{\bar{S}}$ and vice versa (having $\epsilon=0$ as per~\cref{def:swap_complexity}). Again, define the states $\ket{S}$ and $\ket{\bar{S}}$ as 
\begin{align*}
    \ket{S} = \frac{1}{\sqrt{K}} \sum_{x \in S} \ket{x}, \qquad \ket{\bar{S}} = \frac{1}{\sqrt{N-K}} \sum_{x \in \bar{S}} \ket{x},
\end{align*}
where $\bar{S} = \omega\setminus S$. As before, define $\ket{\phi^+} := \frac{1}{\sqrt{2}} \left(\ket{S} + \ket{\bar{S}}\right) $ and $\ket{\phi^-} := \frac{1}{\sqrt{2}} \left(\ket{S} - \ket{\bar{S}}\right)$. Then by~\cref{lem:S_versus_D}, there exists another circuit $C'$ such that for all sets $S \in \mathcal{S}_K$
\begin{align*}
    \abs{\norm{ \Pi_1 C' \ket{\phi^+} \ket{0 \dots 0} }^2_2 -\norm{ \Pi_1 C' \ket{\phi^-} \ket{0 \dots 0}}^2_2} =1,
\end{align*}
which implies that either of the following holds:
\begin{align}
    \norm{ \Pi_1 C' \ket{\phi^+} \ket{0 \dots 0} }^2_2 = 1 \text{ (resp. $0$)} \qquad \text{and} \qquad \norm{ \Pi_1 C' \ket{\phi^-} \ket{0 \dots 0} }^2_2 = 0 \text{ (resp. $1$)}. 
    \label{eq:opt_sp_dist}
\end{align}
In the following, we assume that the first case holds (the argument for the second case in parentheses is identical). 
Now consider two subsets $S_1, S_2 \in \mathcal{S}_K$ and let again  $\ket{\phi^+_i} := \frac{1}{\sqrt{2}} \left(\ket{S_i} + \ket{\bar{S_i}}\right) $ and $\ket{\phi^-} := \frac{1}{\sqrt{2}} \left(\ket{S_i} - \ket{\bar{S_i}}\right)$ for $i \in \{1,2\}$.~\cref{eq:opt_sp_dist} tells us that
\begin{align}
    \norm{ \Pi_1 C' \ket{\phi^+_1} \ket{0 \dots 0} }^2_2 = 1 \qquad \text{and} \qquad \norm{ \Pi_1 C' \ket{\phi^-_2} \ket{0 \dots 0} }^2_2 = 0,
    \label{eq:perfect_disc}
\end{align}
which implies that $\ket{\phi^+_1}$ and $\ket{\phi^-_2}$ can be perfectly discriminated. Writing out the states in full, we have
\begin{align*}
    \ket{\phi^+_1} = \frac{1}{\sqrt{2}} \left( \frac{1}{\sqrt{K}} \sum_{x \in S_1} \ket{x}  + \frac{1}{\sqrt{N-K}} \sum_{x \in \bar{S}_1} \ket{x}   \right) ,
\end{align*}
and
\begin{align*}
    \ket{\phi^-_2} = \frac{1}{\sqrt{2}} \left( \frac{1}{\sqrt{K}} \sum_{x \in S_2} \ket{x}  - \frac{1}{\sqrt{N-K}} \sum_{x \in \bar{S}_2} \ket{x}   \right) .
\end{align*}
We have that the fidelity between $\ket{\phi^+_1}$ and $\ket{\phi^-_2}$ can be computed as
\begin{align*}
    \abs{\bra{\phi^+_1}\ket{\phi^-_2}}^2 &= \frac{1}{4} \abs{ \frac{|S_1 \cap S_2|}{K}  - \frac{|S_1 \cap \bar{S}_2|}{\sqrt{K(N-K)}}   +  \frac{|\bar{S}_1 \cap S_2|}{\sqrt{K(N-K)}}  - \frac{|\bar{S}_1 \cap \bar{S}_2|}{N-K}  }^2\\
    &= \frac{1}{4} \abs{ \frac{|A_1|}{K}  - \frac{|A_2|}{\sqrt{K(N-K)}}   +  \frac{|A_3|}{\sqrt{K(N-K)}}  - \frac{|A_4|}{N-K}  }^2 ,
\end{align*}
where  $A_1 := S_1 \cap S_2, A_2:= S_1 \cap \bar{S}_2, A_3 := \bar{S}_1 \cap S_2 \text{ and } A_4 := \bar{S}_1 \cap \bar{S}_2$.
A simple Venn diagram shows that $\bigcup_j A_j = \omega$ and $\bigcap_j A_j = \emptyset$. Now set $|A_1| = x$, where $x$ should satisfy $\max \{2K - N,0\} \leq x \leq K$. We can then express the sizes of the other sets $A_2$, $A_3$ and $A_4$ as functions of $N$, $K$ and $x$. We have $|A_2| = |S_1 \cap \bar{S}_2| = |S_1 \setminus (S_2 \cap S_1)| = K-x $. A similar argument gives $|A_3| = K-x $. Using the principle of inclusion and exclusion, we have $|S_1 \cup S_2| = |S_1| + |S_2| - A_1 = 2K - x$. Hence, $A_4 = N - |S_1 \cup S_2| = N - 2K + x$. Our expression for the overlap thus becomes
\begin{align}
    \abs{\bra{\phi^+_1}\ket{\phi^-_2}}^2 = \frac{1}{4} \abs{ \frac{x}{K} - \frac{N - 2K + x}{N-K}  }^2 =
    \frac{1}{4} \left(\frac{(2K -N)(K-x)}{K(N-K)} \right)^2 .
\label{eq:overlap_function}
\end{align}
Note that~\cref{eq:overlap_function} becomes $0$ when $K = N/2$ (independent of $x$), as expected. For a fixed $1 \leq K \leq N-1$, we now want to maximize the overlap as a function of $x$, i.e., solve
\begin{align*}
    F_{\text{max}} = \max_{S_1,S_2}  \abs{\bra{\phi^+_1}\ket{\phi^-_2}}^2 =  \max_{\max \{0,2K-N\}  \leq x \leq K} \; \frac{1}{4} \left(\frac{(2K-N)(K-x)}{K(N-K)} \right)^2 .
\end{align*}
The double derivative test gives
\begin{align*}
    \frac{\partial^2}{\partial x^2} \frac{1}{4} \left(\frac{(2K-N)(K-x)}{K(N-K)} \right)^2 =   \frac{(2K-N)^2}{2 (K(N-K))^2} \geq 0,
\end{align*}
which implies that the function is convex in $x$. Therefore, the maximum value is obtained at the extremal points of the interval. For simplicity, assume for now that $K \leq N/2$ so that $ \max \{0,2K-N\} = 0$. For $x=K$, the overlap becomes $0$ (this is because $S_1 = S_2$ and thus $\ket{\phi^{+}_1}$ is orthogonal to $\ket{\phi^{-}_2}$), but for $x=0$ we find that the overlap becomes $\frac{1}{4} \left( \frac{N-2K}{K-N}\right)^2 $.
Note that this expression is valid ($\leq 1$) for $K \leq 3N/4,$ which is true by assumption. For $K>N/2$ the two extremal points are $x=K$, which again gives overlap $0$, and $x=2K-N$, which gives overlap $\frac{1}{4}\left(2 -\frac{N}{K}\right)^2$. The latter is valid ($\leq 1$) for $K\geq N/4$ which is again true by assumption. Combining both, we find that the maximum overlap as a function of $N$ and $K$ can be
\begin{align*}
    \begin{cases}
        \frac{1}{4} \left( \frac{N-2K}{K-N}\right)^2 \qquad &\text{for } 1 \leq  K \leq N/2, \\
        \frac{1}{4}\left(2 -\frac{N}{K}\right)^2 \qquad &\text{for }  N/2 < K \leq N-1.
    \end{cases}
\end{align*}
This is symmetric around $K= N/2$, as substituting $K = N-K'$ with $1 \leq K' \leq N-1$ in $\frac{1}{4}\left(2 -\frac{N}{K}\right)^2$ gives back $ \frac{1}{4} \left( \frac{N-2K'}{K'-N}\right)^2$, which means that our assumption that $K \leq N/2$ can be made without loss of generality. Therefore, for any $K \neq N/2$ we have that $\ket{\phi^{+}_1}$ and $\ket{\phi^-_{2}}$ will not be perfectly orthogonal. However, \cref{eq:perfect_disc} tells us that we can discriminate these two non-orthogonal states perfectly, which is impossible by the Helstrom bound~\cite{helstrom1969quantum}. Hence, there cannot be a perfect swapper for any other case than $K=N/2$. 
\end{proof}

The intuition of why a perfect swapper is only possible when $K=N/2$ follows from the fact that it can be used to construct a distinguisher for the conjugate states, which are generally non-orthogonal unless $K=N/2$.

\subsection{Probabilistic zero-error algorithm for any subset}
\label{sec:zero_erro_alg}

Now that we know a perfect swapper is not possible for $K \neq N/2$, we want to create the ``next best thing'': a swapper that sometimes fails, but only knowingly so. In the following, we show that one can trade success probability to achieve this zero-error requirement. An algorithm satisfying this property is known as a \emph{Las Vegas algorithm}. The idea is to rebalance the Grover diffusion operator $U = 2\dyad{+^n} - \mathbb{I}$ by adding a term proportional to $\mathbb{I}$. To this end, we use one auxiliary qubit and the parameterised gate
\begin{equation}
    W(q) = e^{i \arccos(\sqrt{q}) Y} = \begin{pmatrix} \sqrt{q} & -\sqrt{1-q}\\ \sqrt{1-q} & \sqrt{q} \end{pmatrix} ,
\label{eq:W_gate}
\end{equation}
to construct the following circuit:
\begin{figure}[h!]
    \centering
    \begin{quantikz}[wire types={q,q}]
    \lstick[1]{$\ket{0}$} & \gate{W(q)} & \octrl{1} & \gate{Z^b}  & \gate{W(q)^\dag} & \meter \qw & \rstick[1]{$0$}\setwiretype{c} \\
    \lstick[1]{$\ket{S}$} & \qw \qwbundle{n} & \gate{U} & \qw & \qw & \qw & \qw \rstick[1]{$\ket{\bar{S}}$}
    \end{quantikz}
    \caption{Probabilistic zero-error algorithm to swap $\ket{S}$ of cardinality $K$ to $\ket{\bar{S}}$ of cardinality $N-K$.}
    \label{fig:zero_error_circuit}
\end{figure}

\noindent Here we have that $U$ is controlled on the zero state of the auxiliary qubit, denoted by $C_U$, and $Z$ is the $Z$-gate with some power $b \in \{0,1\}$. The role of $b$ is to allow for both positive and negative deviations from the optimal case of $K=N/2$. 

\begin{theorem}[Zero-error swapper]
Let $S \in \mathcal{S}_K$. Then for any $1 \leq K \leq N-1$, there exists a zero-error (Las Vegas) quantum algorithm that maps $\ket{S}$ to $\ket{\bar{S}}$, with probability
\begin{align*}
    \frac{\min\{K,N-K\} }{\max\{K,N-K\}}.
\end{align*}
\label{thm:zero_error_swapper}
\end{theorem}

\begin{proof}
We will show that for suitable choices of $q$ and $b$, the quantum circuit in~\cref{fig:zero_error_circuit} achieves the desired transformation and corresponding success probability. The state just before the measurement of the auxiliary qubit can be written as
\[
     W(q)^\dag Z^b C_U W(q) \ket{0} \ket{S} = \ket{0} \left(q U + (-1)^b (1-q) \mathbb{I}\right) \ket{S} + \ket{1}\ket{g}
\]
with $\ket{g}$ some state we do not care about. Using that $U = 2\dyad{+^n} - \mathbb{I}$ gives
\[
   q U + (-1)^b (1-q) \mathbb{I} = 2q \dyad{+^n} + \left( (-1)^b (1 - q) - q \right) \mathbb{I}.
\]
We will use
\[
    \dyad{+^n} \ket{S} = \frac{K}{N} \ket{S} + \sqrt{\frac{K(N-K)}{N^2}} \ket{\bar{S}},
\]
to obtain
\begin{align}
    \left(q U + (-1)^b (1-q) \mathbb{I}\right)\ket{S} =\left( 2q \frac{K}{N} + (-1)^b (1 - q) - q \right) \ket{S}
    + 2q \sqrt{\frac{K(N-K)}{N^2}} \ket{\bar{S}}.
    \label{eq:comp_S_set_zero}
\end{align}
To ensure that $\left(q U + (-1)^b (1-q)\mathbb{I}\right)\ket{S}$ is proportional to $\ket{\bar{S}}$, we set the coefficient of $\ket{S}$ to zero:
\begin{equation*}
    2q \frac{K}{N} + (-1)^b (1 - q) - q = 0.
\end{equation*}
This gives two valid solutions:
\begin{itemize}
    \item If $b = 0$, the equation becomes $2q \frac{K}{N} + 1 - 2q = 0$, which gives
    \[
        q = \frac{1}{2(1 - K/N)}.
    \]
    \item If $b = 1$, the equation becomes $2q \frac{K}{N} - 1 = 0$, which gives
    \[
        q = \frac{N}{2K}.
    \]
\end{itemize}
It remains to check when these choices yield a valid success probability (or equivalently, for what values of $q$ the operator $W(q)$ is indeed unitary), i.e., when 
\[
\norm{\left(q U + (-1)^b (1-q)\mathbb{I}\right)\ket{S}}^2 \leq 1
\]
Using that in~\cref{eq:comp_S_set_zero} we made the amplitude on $\ket{S}$ zero and the $\ket{\bar{S}}$-part only depends implicitly on $b$ via $q$, we obtain
\[
    \left\|\left(q U + (-1)^b (1-q)\mathbb{I}\right)\ket{S} \right\|^2 = 4q^2 \frac{K(N-K)}{N^2}.
\]
We find that for the value of $q$ corresponding to $b = 0$, the norm squared is $\frac{K}{N - K}$, which is at most $1$ when $K \leq N/2$. When $b = 1$, we obtain $\frac{N - K}{K}$, which is at most $1$ when $K \geq N/2$. To summarise, we set
\begin{equation*}
    \begin{cases}
    b = 0 , \; q = \frac{1}{2(1 - K/N)} &\text{ if } K < \frac{N}{2}, \\
    b = 1 , \; q = \frac{N}{2K} &\text{ if } K \geq \frac{N}{2},
    \end{cases}
\end{equation*}
in~\cref{fig:zero_error_circuit}. The overall success probability is given by the postselection success probabilities in both cases (the norms squared), which can be captured in a single expression as
\begin{align*}
\frac{\min\{K,N-K\} }{\max\{K,N-K\}}.
\end{align*}
\end{proof}

\subsection{Zero error: optimality, auxiliary qubits and multiple samples}
\label{sec:zero_error}

In this section we will show that the algorithm from~\cref{sec:zero_erro_alg} is in fact optimal when considering its restricted setting: we assume that there is some designated flag qubit, which when measured to be in the $\ket{0}$ state indicates that the state has been successfully created to ensure zero error (putting the flag being in $\ket{0}$ is without loss of generality).  Except for this flag qubit, we require that the circuit uses no other extra auxiliary qubits. We prove the following proposition.

\begin{proposition} 
    For all $1 \leq K \leq N - 1$, under the zero-error condition and the use of only one extra ancilla (used as a flag to indicate that the operation was successful),~\cref{thm:zero_error_swapper} is optimal in terms of achievable success probability for swapping $\ket{S}$ to $\ket{\bar{S}}$.
\label{prop:optimality}
\end{proposition}

\begin{proof}
For cardinality $K$, let $C_K$ be an arbitrary circuit which probabilistically swaps $\ket{S}$ for subsets $S \in \mathcal{S}_K$ to $\ket{\bar{S}}$. For full generality, we allow the preparation of $\ket{\bar{S}}$ up to an arbitrary global phase factor $e^{i \theta}$, $\theta \in [0, 2\pi]$. Let $0\leq \epsilon \leq 1$. To achieve success probability $1-\epsilon$, $C_K$ should be able to perform a mapping of the form
\begin{align}
    C_K \ket{0} \ket{S_j}  = e^{i\theta_j} \sqrt{1-\epsilon}  \ket{0} \ket{\bar{S}_j} + \sqrt{\epsilon} \ket{1}\ket{g_j}
\label{eq:generic_swapped_state}
\end{align}
for some arbitrary subset $S_j \in \mathcal{S}_K$, phase $\theta_j \in [0,2\pi]$, and some garbage state $\ket{g_j}$ (which can absorb its own phase factor). Now take two subsets $S_1$ and $S_2$ defined as $S_1 = [K]$ and $S_2 = \omega\setminus[N-K]$ where, with a slight abuse of notation, we write $[K]$ to indicate the $K$ first strings in $\omega$ under lexicographical ordering. We distinguish two cases:
\begin{itemize}
\item $K \geq N/2$. In this case, we have that the inner product between input states $\ket{S_1}$ and $\ket{S_2}$ is given by (recall that each $\ket{S_j}$ only has positive amplitudes on its computational basis states)
\begin{align*}
    \bra{S_1}\ket{S_2}= \frac{2K-N}{K} .
\end{align*}
The overlap between output states is given by
\begin{align*}
     \abs{\bra{0}\bra{S_1} C_K^\dagger C_K \ket{0}\ket{S_2}} = \abs{(1-\epsilon) e^{i(\theta_{2}-\theta_{1})} \bra{\bar{S}_1}\ket{\bar{S}_2} + \epsilon \bra{g_1}\ket{g_2}}.
\end{align*}
By our construction of the subsets we have that $\bra{\bar{S}_1}\ket{\bar{S}_2} = 0$. Since the overlap must be conserved under unitary transformations, we must have
\begin{align*}
    \frac{2K-N}{K}  = \epsilon \abs{\bra{g_1}\ket{g_2}}.
\end{align*}
To find an upper bound on $1-\epsilon$, we can solve the following optimization problem:
\begin{align*}
    \max_{\epsilon, \abs{\bra{g_1}\ket{g_2}}}\quad & 1 - \epsilon \\
    \text{s.t.} \quad & \frac{2K - N}{K} = \epsilon \abs{\bra{g_1}\ket{g_2}}, \\
    & \epsilon \in [0,1], \\
    & \abs{\bra{g_1}\ket{g_2}} \in [0,1].
\end{align*}
This gives us 
\begin{align*}
    1-\epsilon \leq 1-\frac{2K-N}{K} = \frac{N - K}{K}.
\end{align*}

\item $K \leq N/2$. Let us first make the observation that $C_K^{\dagger}$ gives us an approximate swapper to go from $\ket{\bar{S}}$ to $\ket{S}$, where the cardinality of $\bar{S}$ is $K' = N-K$ with $K' \geq N/2$. In particular
\begin{align*}
    \abs{ \bra{0} \bra{S_j}  C_K^\dagger \ket{0}\ket{\bar{S}_j} } = \abs{ \bra{0}\bra{\bar{S}_j} C_K \ket{0}\ket{S_j}} =  \sqrt{1-\epsilon} ,
\end{align*}
by ~\cref{eq:generic_swapped_state}.
This means we can write the action of $C^\dagger_K$ on $\ket{0}\ket{\bar{S}_j}$ to be of the form
\begin{align*}
    C_K^\dagger \ket{0} \ket{\bar{S}_j}  = e^{i \bar{\theta}_j} \sqrt{1-\epsilon}\ket{0} \ket{S_j} + \sqrt{\epsilon} \ket{1}\ket{\bar{g}_j},
\end{align*}
again for some garbage state $\ket{\bar{g}_j}$ and some phase $\bar{\theta}_j \in [0, 2\pi]$. 
We can now follow the same argument as the previous case, but looking at the inner product between the complement states. By our construction of the subsets we have
\begin{align*}
    \bra{\bar{S}_1} \ket{\bar{S}_2} = \frac{2K' -N}{K'} = \frac{N-2K}{N-K} . 
\end{align*}
But this must be equal to  
\begin{align*}
    \abs{\bra{0}\bra{\bar{S}_1} C_K C_K^\dagger \ket{0}\ket{\bar{S}_2}} =     \abs{(1-\epsilon) e^{i(\bar{\theta}_{2}-\bar{\theta}_{1})} \bra{S_1}\ket{S_2} + \epsilon \bra{\bar{g}_1}\ket{\bar{g}_2}},
\end{align*}
where $\bra{S_1} \ket{S_2} = 0$.  Maximizing $1 -\epsilon$ as done in the previous case, but
using $K'$ instead of $K$, we find
\begin{align*}
    1-\epsilon \leq \frac{N}{K'} - 1 = \frac{K}{N-K}.
\end{align*}
\end{itemize}
The optimal probability of the two cases can be summarized as $\frac{\min\{K,N-K\} }{\max\{K,N-K\}}$, showing that under these conditions the algorithm from~\cref{sec:zero_erro_alg} is optimal.
\end{proof}

Note that the above proof is an arguably easier proof of the result in~\cref{sec:impossibility}, as it only relies on the conservation of the absolute value of the inner product under unitary transformations. There are several more observations we can make in relation to the proof of~\cref{prop:optimality}.

\paragraph{Adding extra auxiliary qubits.}
The proof of~\cref{prop:optimality} shows that, under the condition of having only one additional ancilla to work as flag qubit, the algorithm of~\cref{sec:zero_erro_alg} achieves the optimal success probability. What happens when extra auxiliary qubits are allowed (initialized in $\ket{0 \dots 0}$), but we still require a single flag qubit to make the algorithm zero-error (yet still probabilistic)? 

In the $K \geq N/2$ case, it is easy to show that the above proof still holds. For simplicity, we remove the relaxation that the state can be prepared up to an arbitrary global phase, as it does not change the argument. We can write the action of the swapping circuit in this case as
\begin{align*}
    C_K \ket{0} \ket{S_j} \ket{0 \dots 0}  = \sqrt{1-\epsilon}\ket{0} \ket{\bar{S}_j} \ket{h_j} + \sqrt{\epsilon} \ket{1}\ket{G_j},
\end{align*}
for some garbage states $\ket{h_j}$ and $\ket{G_j}$, so we still have that the inner product of the output state after applying $C_K$ to two input subset states $\ket{S_1}$ and $\ket{S_2}$ is given by
\begin{align*}
    (1-\epsilon) \bra{\bar{S}_1}\ket{\bar{S}_2}\bra{h_1}\ket{h_2} + \epsilon \bra{G_1}\ket{G_2} = \epsilon \bra{G_1}\ket{G_2},
\end{align*}
again using that $\bra{\bar{S}_1}\ket{\bar{S}_2} = 0$ when $K \geq N/2$.

However, the argument for $K\leq N/2$ no longer works since $\ket{h_1}$ and $\ket{h_2}$ might be two completely different states. This means that to swap backwards using $C_K^\dagger$ one might have to use a different initial state for each different input $\bar{S}$.

An easy example that contradicts the bound in the case when extra auxiliary qubits can be used is the case of $K=1$ (where the bound gives a maximum success probability very close to $0$). An algorithm that only depends on $K$ and perfectly makes $\ket{\bar{S}}$ is easily given: (i) read out the only element $x \in S$, (ii) create the uniform superposition, (iii) use an additional flag qubit to mark the state $\ket{S}$, which can be done by applying a marking operation that marks all basis states conditioned on not being equal to $\ket{x}$, and (iv) use exact Grover's algorithm to create $\ket{\bar{S}}$~\cite{zalka1999grover}. This algorithm clearly needs additional auxiliary qubits, as it cannot both have a uniform superposition on $n$ qubits and a flag controlled on $\ket{x}$ (which also needs to be stored in a $n$-qubit register). Note that in this case the problem is also classically trivial, because after observing a single sample $x$ one has perfectly learned $\bar{S} = \omega \setminus \{ x\}$.

\paragraph{Multiple copies and no additional auxiliary qubits.}
The above argument also shows that a joint measurement on multiple copies does not help if one does not have additional auxiliary qubits. That is, suppose one again requires a probabilistic zero-error algorithm that is able to perform the mapping
\begin{align*}
    C_K \ket{0} \ket{S_j}^{\otimes k}  = \sqrt{1-\epsilon}\ket{0} \ket{\bar{S}_j} \ket{h_j} + \sqrt{\epsilon} \ket{1}\ket{G_j},
\end{align*}
using $k$ copies of the input state. Considering $K\geq N/2$, the same argument as before gives us a maximum success probability of 
\begin{align*}
    1-\epsilon \leq 1- \left(\frac{2K-N}{K} \right)^k,
\end{align*}
which is the same success probability one would achieve by repeating our algorithm $k$ times, resetting the flag qubit to $\ket{0}$ whenever $\ket{1}$ is measured and the protocol failed. We currently do not know whether a joint measurement on multiple copies and the use of extra auxiliary qubits would yield an algorithm with higher success probability.

\section{Classical samples: lower and upper bounds}
\label{sec:classical}

\subsection{Index query model} 
\label{sec:worst_case_query}

We will start by adopting a slightly more powerful classical setting, as it allows for easier proofs and it will be useful later when we construct $S$ from pseudorandom permutations, which naturally assumes a query complexity setting. Instead of assuming that the algorithm is given samples from $S$ at random, we assume that it has access to an oracle for which it can query elements of $S$ (assuming arbitrary unknown ordering on each of the possible input sets $S$). That is, it has access to an oracle $O_{\textup{index}}$ such that $O_{\textup{index}}(i) = y$ with $y$ the $i$th element from $S$. Clearly, this is a stronger access model, as the player can generate the random samples by querying $O_{\textup{index}}$ for randomly generated indices.

Intuitively, one would immediately expect that complement sampling with classical samples is intractable: when there is no structure to the set $S$ (and therefore no structure to $\bar{S}$), the best strategy seems to be to just gather (or in this case, query) as many distinct samples from $S$ and just output anything which is not in the set of collected samples. We will now formalise this intuition and show that it is indeed the correct way to look at the problem. We will use Yao's minimax principle~\cite{yao1977probabilistic} to prove our lower bound. For a fixed input size $N$ and cardinality $K$, let $F(S,y)$ be a function describing the relation for a subset $S$ from a subset family $\mathcal{S}_{K}$ and candidate element $y$ such that
\begin{align*}
    F(S,y) = \begin{cases}
        1 \text{ if } y \in \bar{S}\\
        0 \text{ otherwise}.
    \end{cases}
\end{align*}
Let $A$ be the set of all possible deterministic algorithms allowed to make $T$ queries to the index oracle for $S$.  Write $P(a,S) = F(S,a(S))$, which means that $P = 1$ if and only if algorithm $a \in A$ on input $S$ outputs any $y \in \bar{S}$. One can view $P(a,S)$ as a $|A| \times |\mathcal{S}_{K}|$- matrix where the rows label the different choices of deterministic query algorithms, the columns label the different instances of $S$, and the corresponding entry is $1$ only if the algorithm provides an output satisfying the relation. By the minimax theorem, we have
\begin{align*}
    \min_{\mu} \max_{a} e_a^{\intercal} P \mu = \max_{\rho} \min_{S} \rho^{\intercal} P e_S,
\end{align*}
where $e_a$ (resp.~$e_S$) is a $|A|$-dimensional (resp. $|\mathcal{S}_{K}|$-dimensional) unit vector, and $\rho$ (resp.~$\mu$) is a $|A|$-dimensional (resp.~$|\mathcal{S}_{K}|$-dimensional) vector of non-negative reals that sum to one. Since $\mu$ is a probability distribution, we have that the left-hand side describes the probability that the best deterministic $T$-query algorithm is correct on the hardest distribution over inputs $\mu$. On the right-hand side, $\rho^{\intercal} P e_S $ is the success probability on input $S$ achieved by the randomised algorithm given by probability distribution $\rho$ over deterministic algorithms. Hence, the right-hand side gives the highest worst-case success probability achievable by randomised $T$-query algorithms. Since this holds for all $T$, we have that 
\begin{align*}
    R_{\epsilon} (F) = \max_{\mu} D_{\epsilon}^{\mu} (F),
\end{align*}
where $ R_{\epsilon}$ is the worst-case randomized query complexity, achieving success probability $1-\epsilon$ on all inputs, and  $D_{\epsilon}^{\mu}$ the deterministic query complexity on succeeding on a $(1-\epsilon)$-fraction over all inputs weighted by $\mu$. Note that we now have a maximization over input distributions $\mu$, since we have moved from expressing the smallest success probability given a fixed maximum number of queries $T$ to the largest minimum number of queries $T$ given a fixed fraction of inputs on which we have to be correct. For any fixed candidate distribution $\mu$, we always have
\begin{align*}
    R_{\epsilon} (F) \geq D_{\epsilon}^{\mu} (F).
\end{align*}

We first prove the following lemma, showing that when starting with the uniform distribution over all families of subsets $\mathcal{S}_{K}$ (and thus also a uniform distribution over complementary subsets $\bar{S}$), learning entries of $S$ will leave the resulting conditional distribution over the family of complementary subsets that do not include these elements still uniform.

\begin{lemma}
    Let $\mu_{\mathcal{S}_K}$ be the uniform distribution over all $S \in \mathcal{S}_K$. Let $X \in \mathcal{S}_K$ be sampled according to $\mu_{\mathcal{S}_K}$, and define $\bar{X} = \omega \setminus X$. Fix a set $Q = \{x_{1},\dots, x_{q}\}$ containing $q \leq K$ distinct elements from $\omega$. Then, conditioned on $Q \subseteq X$, the distribution of $\bar{X}$ is uniform over $\mathcal{S}_{N-K}^Q$, the family of all subsets of size $N-K$ that do not contain any element from $Q$.
    \label{lem:bayes}
\end{lemma}

\begin{proof}
    By Bayes' rule,
    \begin{align*}
        \Pr[\bar{X} = \bar{S} \mid x_{1},\dots, x_{q} \in X] &= \frac{\Pr[\bar{X} = \bar{S} \land x_{1},\dots, x_{q} \in X]}{\Pr[x_{1},\dots, x_{q} \in X]} .
    \end{align*}
    Using that
    \begin{align*}
        \Pr[x_{1},\dots, x_{q} \in X] = \frac{\binom{N-q}{K-q}}{\binom{N}{K}},
    \end{align*}
    and the fact that
    \begin{align*}
       \Pr[\bar{X} = \bar{S} \land x_{1},\dots, x_{q} \in X]
        &= \begin{cases}
            \binom{N}{N-K}^{-1} & \text{if } x_{1},\dots, x_{q} \notin \bar{S},\\
            0 & \text{otherwise},
        \end{cases}
    \end{align*}
    we obtain
    \begin{align*}
        \Pr[\bar{X} = \bar{S} \mid x_{1},\dots, x_{q} \in X] &=  \begin{cases}
            \binom{N-q}{N-K}^{-1} & \text{if } x_{1},\dots, x_{q} \notin \bar{S},\\
            0 & \text{otherwise},
        \end{cases}
    \end{align*}
    which is uniform over all $\bar{S} \in \mathcal{S}^Q_{N-K}$.
\end{proof}

We can now use Yao's principle to derive our lower bound.

\begin{theorem} 
Let $\delta \in [0,\frac{1}{2}]$ and $1 \leq K \leq N-1$. Then we have that any randomized algorithm that produces a sample from $\bar{S}$ with probability $\geq 1/2 + \delta$ on all inputs $S$ needs to make at least
\begin{align*}
         N-\frac{2 (N-K)}{2 \delta + 1}.
\end{align*}
queries to the index oracle for $S$. 
\label{thm:LB_queries}
\end{theorem}

\begin{proof}
    Consider $\mu_{\mathcal{S}_K}$ to be the uniform distribution over all possible inputs $S$ coming from the family $\mathcal{S}_K$. Let $X \in \mathcal{S}_K$ be sampled according to $\mu_{\mathcal{S}_K}$. Let $A$ be a deterministic algorithm with access to the index oracle $O_\text{index}$.  We can assume that no element of $X$ is queried twice by $A$, as there would be no benefit in doing so. Then after $q$ queries, the algorithm has seen $q$ elements of $X$. Let $Q = \{x_1,\dots,x_q\}$ be the set of all strings $x_j \in \{0,1\}^n$ that were observed, and write $\mathcal{S}^Q_{N-K}$ for the family of all sets $\bar{S}$ that do not contain the strings $x_j \in Q$. By~\cref{lem:bayes}, we have that the conditional distribution $\mu_{\mathcal{S}^{Q}_{N-K}}$ over $\mathcal{S}^{Q}_{N-K}$ is still uniform over these elements. The cardinality of the set $\mathcal{S}^{Q}_{N-K}$ is then given by
    \begin{align*}
        \left|\mathcal{S}^{Q}_{N-K}\right| = \binom{N-q}{N-K}.
    \end{align*}
    So for any output $y$, the probability that it is part of the $\bar{S}'$ for some $\bar{S}' \in \mathcal{S}^{Q}_{N-K}$ is then given by 
    \begin{align}
        \frac{\binom{N-q-1}{N-K-1}}{\binom{N-q}{N-K}} = \frac{N-K}{N-q},
        \label{eq:y_correct_prob}
    \end{align}
    which holds independently of the strings that are in $Q$ (except for how many of them there are). Hence, $A$ can only be correct on any such fraction of the inputs. By Yao's principle, this is also the highest success probability a $q$-query randomized algorithm can achieve on the worst-case input. Then, to be correct with probability at least $\frac{1}{2} + \delta$ we require
    $
        (N-K)/(N-q) \geq \frac{1}{2} + \delta$,
    which is satisfied when 
    \begin{align*}
        q \geq N-\frac{2 (N-K)}{2 \delta + 1}.
    \end{align*}
\end{proof}

We proceed by giving a matching upper bound, showing that~\cref{thm:LB_queries} is tight.

\begin{proposition} 
    For any $\delta \in [0, \frac{1}{2}]$, for any $S$ there exists a randomized algorithm which makes $N-\frac{2 (N-K)}{2 \delta + 1}$ queries to $O_{\textup{index}}$ and outputs an element $y \in \bar{S}$ with probability $\frac{1}{2} + \delta$.
\end{proposition}

\begin{proof}
    Again, let $Q = \{x_1,\dots,x_{q}\}$ be the set of strings $x_j \in S$ that are observed after $q$ queries. The algorithm will now simply sample a uniformly random $y$ from the set $\omega \setminus Q$. Given any $S$, the probability that this $y$ is from $\bar{S}$ is then given by
    \begin{align*}
        \Pr \left[y\in \bar{S}\right] = \frac{N-K}{N-q} =  \frac{N-K}{N-\left(N-\frac{2 (N-K)}{2 \delta + 1}\right)} = \frac{1}{2} + \delta.
    \end{align*}
\end{proof}

Note that if $K = \Omega(N)$ and $\delta = \Omega(1)$, then we have $q \geq \Omega(N)$. 

\subsection{Exact sample complexity bounds}
\label{sec:exact_sampl_complexity}

In the above, we used an index query model because it greatly simplifies the analysis and provides lower bounds for the sample complexity setting. However, it is possible to translate these bounds to the sample complexity setting by computing the probability of observing $q$ unique samples after drawing $d$ samples. Since the results of~\cref{sec:worst_case_query} hold in terms of \emph{unique} observations of elements from  $S$, we can exactly compute matching upper and lower bounds on the sample complexity.

Sampling from a subset of size $K$, the probability of getting $q \leq d$ unique samples from $d$ draws (with replacement, each with equal probability) is given by (see for example~\cite{mendelson2016distribution} for a derivation)
\begin{align*}
    \Pr[X = q] = \frac{K \stirling{d}{q} }{(K-q)! K^d} ,
\end{align*}
where $\stirling{d}{q}$ is the Stirling number of the second kind, representing the number of ways to partition $d$ objects into $q$ non-empty subsets. Combining this with~\cref{eq:y_correct_prob}, which gives the probability of outputting a successful sample conditioned on observing $q$ unique elements, gives a lower bound on the overall success probability of 
\begin{align*}
    \sum_{q=0}^{q = d} \frac{N \stirling{d}{q} }{(K-q)! K^d} \frac{N-K}{N-q}.
\end{align*}
This in principle gives the exact expression and has a matching upper bound by the same argument as the previous subsection. However, since the expression is cumbersome to work with and we are interested in proving asymptotic lower bounds, we will continue with the query model in the following sections to consider more practically relevant notions of hardness of the problem.

\subsection{Average-case lower bounds}
A consequence of~\cref{thm:LB_queries} is that the problem is hard on average (with respect to the uniform distribution), as a worst-to-average case reduction with respect to the uniform distribution can easily be constructed. We will use that applying a random permutation on $N$ elements as an operation to any fixed subset of $K$ out of the $N$ elements will result in a uniformly random subset of $K$ elements.

\begin{lemma}
Suppose $\sigma$ is drawn uniformly at random from the set of all permutations on $\{0,1\}^n$. Then, for any fixed subset $S \in \mathcal{S}_K$, the associated subset 
\[
S' = \{ x \in \{0,1\}^n : \sigma^{-1}(x) \in S \} = \sigma(S).
\]
is uniformly drawn from the family $\mathcal{S}_K$.
\label{lem:random_subset_random_P}
\end{lemma}
\begin{proof}
We have that $S'$ is the image of $S$ under the permutation $\sigma$. We want to understand the distribution of $S'$ when $\sigma$ is chosen uniformly at random. Fix any subset $T \in \mathcal{S}_K$. We count how many permutations $\sigma$ satisfy $\sigma(S) = T$. To construct such a permutation, we must:
\begin{itemize}
    \item define a bijection from $S$ to $T$, of which there are $K!$ choices;
    \item define a bijection from $\overline{S}$ to $\overline{T}$, of which there are $(N - K)!$ choices.
\end{itemize}
So the total number of permutations $P$ such that $\sigma(S) = T$ is
\[
|\{\sigma : \sigma(S) = T\}| = K! \cdot (N - K)!,
\]
which is the same for all $T \in \mathcal{S}_K$. Since $\sigma$ is chosen uniformly at random, this means that $S' = \sigma(S)$ is uniformly distributed over $\mathcal{S}_K$. In particular,
\[
\Pr[S' = T] = \frac{K! \cdot (N - K)!}{N!} = \frac{1}{\binom{N}{K}} = \frac{1}{|\mathcal{S}_K|},
\]
as required.
\end{proof}
We now show that the existence of a $q$-query randomized algorithm $A$ that achieves a certain guaranteed expected success probability averaged over its internal randomness and the input distribution, implies the existence of a randomized $q$-query algorithm $B$ that achieves the same guaranteed expected success probability on every input. This yields the desired conclusion that the complement sampling problem is average-case just as hard as worst-case.

\begin{corollary} 
    Let $K = N/2$, $N= 2^n$ for $n \in \mathbb{N}$, and $\delta$ be any function such that $\delta(n) \in [0,\frac{1}{2}]$ for all $n \in \mathbb{N}$. Suppose we sample some $S \in \mathcal{S}_{K}$ uniformly at random. Then any classical randomized algorithm $A$ which makes $q$ queries to $O_{\textup{index}}$ and satisfies
    \begin{align}
        \Pr\left[A \textup{ outputs } y \in \bar{S}\right] \geq \frac{1}{2} + \delta
        \label{eq:cor_eq}
    \end{align}
    has
    \begin{align*}
        q \geq N \left( 1 - \frac{1}{2\delta + 1} \right),
    \end{align*}
    where the probability in~\cref{eq:cor_eq} is taken over $S$ and the randomness of $A$.
\label{cor:av_case}
\end{corollary}

\begin{proof}
Let $\sigma$ be a permutation on bit strings of length $n$ chosen uniformly at random, and write $\sigma^{-1}$ for its inverse. By~\cref{lem:random_subset_random_P}, starting from any $S \in \mathcal{S}_{N/2}$, $\bar{S} = \omega \setminus S$, we have that the subset $S'$ and its complement $\bar{S}' = \omega\setminus S'$ given by
\begin{align*}
    S' = \{ x'\mid x' = \sigma(x), x\in S \}, \qquad \bar{S}' = \{ y '\mid y' = \sigma(y), y \in \bar{S} \},
\end{align*}
are uniformly at random from all possible $S \in \mathcal{S}_{N/2}$. Since we are only interested in query (or sample) complexity, we do not care that the time complexity of implementing this permutation $\sigma$, nor its inverse $\sigma^{-1}$, generally scales exponentially in $n$ (recall $n = \log N)$. Now suppose there exists a $q$-query classical algorithm $A$ that succeeds with expected probability at least $\frac{1}{2} + \delta$ over this distribution of inputs. We can then use $A$ to create a $q$-query algorithm $B$ that works with a certain success probability on all inputs in the following way:
\begin{enumerate}
    \item Pick a permutation $\sigma$ uniformly at random and let $\sigma^{-1}$ be its inverse.
    \item Let $O_{\textup{index}}'(i) = \sigma(O_{\textup{index}}(i))$. Run algorithm $A$ with $O_{\textup{index}}'(i)$, resulting in some $y'$.
    \item Return $y = \sigma^{-1} (y')$.
\end{enumerate}
We then have that $B$ succeeds on any input $S$ with probability at least $\frac{1}{2} + \delta$  as well, since
\begin{align*}
    \Pr\left[B \text{ outputs } y \in \bar{S}\right] &= \mathbb{E}_{S'} \left[ \Pr\left[A \text{ outputs } y' \in \bar{S}'\right] \right] \geq \frac{1}{2} + \delta,
\end{align*}
using that the expectation of a probability is again a probability. By~\cref{thm:LB_queries}, we then need at least
\begin{align*}
    q &\geq N \left( 1 - \frac{1}{2\delta + 1} \right)
\end{align*}
queries to $O_{\textup{index}}$, completing the proof.
\end{proof}

\subsection{Hardness for strong pseudorandom subsets}

We will now prove our strongest hardness result: we will show that under the existence of one-way functions, there should be families of subsets that are easy to generate and verify, but still hard to perform complement sampling on classically. To achieve this, we will use strong pseudorandom permutations.  Following the convention in cryptography, we say that a function $f$ is negligible ($f(n) = \textup{negl}(n) $), if for every constant $c$, we have $f(n) = o\left( \frac{1}{n^c} \right) $. 

\begin{definition}[Strong pseudorandom permutations]
    For some $n \in \mathbb{N}$, let $P : \{0,1\}^n \cross \{0,1\}^{\lambda} \rightarrow \{0,1\}^n$ be an efficient keyed permutation. We say that $P$ is a strong pseudorandom permutation if for all probabilistic polynomial-time distinguishers $D$ there exists a negligible function $\textup{negl}(n)$ such that
    \begin{align*}
        \abs{\Pr\left[ D^{P_k (\cdot),P^{-1}_k (\cdot)}(1^n)\right]-\Pr\left[ D^{\sigma (\cdot),\sigma^{-1}(\cdot)}(1^n)\right]} \leq \textup{negl}(\lambda),
    \end{align*}
    where $k \leftarrow \{0,1\}^\lambda$ is chosen uniformly at random and $\sigma$ is chosen uniformly at random from the set of permutations on $n$-bit strings.
\end{definition}

It is known that strong pseudorandom permutations can be constructed from pseudorandom permutations~\cite{luby1988how}, which exist if and only if one-way functions exist~\cite{katz2020introduction}.

We will show the hardness result for cardinality $K = N/2$, as it forms the basis of our proposed advantage experiment. However, the argument can be easily extended to any $K$ superpolynomial in $N$.

\begin{theorem} 
    For some $n \in \mathbb{N}$, let $P : \{0,1\}^n \times \{0,1\}^{\lambda} \rightarrow \{0,1\}^n$ be a strong pseudorandom permutation with security parameter $\lambda = n$, and $P^{-1}$ its inverse. Given a key $k \in \{0,1\}^{\lambda}$, let $S$ be the image of $\{0ik : i \in \{0,1\}^{n-1} \}$ under $P$, and $\bar{S}$ be the image of $\{1ik : i\in \{0,1\}^{n-1} \}$ under $P$. Let $O_{\textup{index}} : [2^{n-1}] \rightarrow \{0,1\}^{n}$ be the oracle that on input $i$ returns the $i$-th element in $S$.  Picking a key uniformly at random, for all polynomial-time algorithms $A$ that make a polynomial number of queries to $O_{\textup{index}}$ it must hold that
    \begin{align*}
        \Pr \left[A^{O_{\textup{index}}} \textup{ outputs a } y \in \bar{S}\right] \leq \frac{1}{2} + \textup{negl}(n).
    \end{align*}
\label{thm:hardness_PRP_S}
\end{theorem}

\begin{proof}
We will use a proof by contradiction. Suppose that there exists a polynomial-time algorithm $A$ with query access to some $O_{\textup{index}}$, which, when picking a key uniformly at random to construct $S$, outputs a $y \in \bar{S}$ with probability $\geq \frac{1}{2}+1/\poly(n)$. We show that if such an $A$ existed it would imply a distinguisher between a strong pseudorandom permutation and truly random one, violating the assumption that $S$ was generated by a strong pseudorandom permutation. First, note that defining $S$ to be the set of all outputs of a uniformly random permutation for which the preimage has `$0$' as the first bit ensures that $S$ is chosen uniformly at random from $\mathcal{S}_{N/2}$ (\cref{lem:random_subset_random_P}). Second, note that any given $P$ already acts as $O_{\textup{index}} $ when applied to strings from $\{0ik : i \in \{0,1\}^{n-1} \}$. Let $F$ be a samplable family of either (i) permutations or (ii) strong pseudorandom permutations.
Then, we can design a distinguisher using $A^{O_{\textup{index}}}$ as a subroutine in the following way:
\begin{enumerate}
    \item We sample a permutation $f \in F$, which is either strong pseudorandom or truly random, and construct $O_{\textup{index}}$ from $f$. 
    \item We run $A^{O_{\textup{index}}}$ to obtain a candidate string $\hat{y}$, supposedly from the complement.
    \item We check whether $f^{-1} (\hat{y})$ has the first bit equal to `$1$'. If this is the case, we return \checkmark. If not, we return \xmark.
\end{enumerate}
Let $\delta = 1/\sqrt{N}$, $N = 2^n$, such that $\delta \in \textup{negl}(n)$. Then,~\Cref{cor:av_case} readily implies that whenever $q = \poly(n)$ and $S$ is uniformly random (i.e., $f$ comes from a family of truly random permutation), then any $q$-query classical algorithm (even without the requirement that it runs in polynomial time) succeeds with probability smaller than $\frac{1}{2} + \delta = \frac{1}{2}+\textup{negl}(n)$. However, if $f$ comes from a family of strong pseudorandom permutations then by assumption it must hold that we observe a \checkmark with probability $\frac{1}{2} + 1/\poly(n)$. Hence, under this assumption, we can distinguish strong pseudorandom permutations from truly random ones, by repeating the above distinguisher a polynomial number of times to detect the small polynomial bias towards returning \checkmark in the pseudorandom case. Since this is not possible by the definition of strong pseudorandomness, such an $A$ cannot exist. Thus the statement must be true.
\end{proof}

The above argument does not extend to pseudorandom permutations that are not strong, as we need to use $P^{-1}$ to check whether the generated string is indeed from $\bar{S}$.

\subsection{Circuit lower bounds for uniform samplers}
\label{sec:classical_circuit_lower_bound}

Finally, we conclude our classical hardness results by showing that we can even prove an exponential \emph{circuit lower bound} on any sampler (quantum or classical) that uses only a polynomial number of classical samples and produces a uniformly random sample from $\bar{S}$ with probability $1$. Unlike the previous lower bounds, which hold even when the uniformly random criterion is removed, this bound relies crucially on the requirement that every element from $\bar{S}$ has a non-zero probability to be produced by the sampler. The key idea is to use Kolmogorov complexity of binary strings to arrive at a circuit lower bound.

\begin{definition} 
    The Kolmogorov complexity $\mathcal{K}(x)$ of any binary string $x \in \{0,1\}^n$ is the length of the shortest computer program $x^*$ that can produce this string on the Universal Turing Machine and then halts.
    \label{def:K_complexity}
\end{definition}

It is easy to show that there must exist incompressible strings (so with $\mathcal{K}(x) \geq |x|$) for every input length $n$: there are $2^n$ binary strings of length $n$ but only $2^n-1$ binary strings of length strictly less than $n$ to describe the program. 

We will formulate our bound in terms of the circuit complexity of a \emph{Boolean circuit}. A Boolean circuit is a directed acyclic graph consisting of logic gates such as AND, OR, and NOT, wired together to compute a Boolean function. The size~$s$ of the circuit is defined as the total number of gates (the total number of vertices in the graph). Any Boolean circuit of size~$s$ can be described using at most $\mathcal{O}(s \log s)$ bits: each gate can be encoded by its type and the indices of its input wires, each requiring $\mathcal{O}(\log s)$ bits. Consequently, every such circuit can be converted into a binary string~$x^*$ of length $|x^*| = \mathcal{O}(s \log s)$, which serves as a program for the universal Turing machine that simulates the circuit. Hence, if a string~$x$ has Kolmogorov complexity $\mathcal{K}(x)$ as defined in~\cref{def:K_complexity}, then any circuit computing~$x$ must have size at least $
s = \Omega\left(\mathcal{K}(x)/\log \mathcal{K}(x)\right)$.

We will use this to show a circuit lower bound on any classical sampler that uses only a polynomial number of samples and is able to sample uniformly at random from the complementary subset $\bar{S}$.

\begin{theorem}[Circuit complexity lower bound with classical samples] 
    There exists an $S \in \mathcal{S}_K$, with $K=N/2$, such that any sampler, which given as an input $l = \poly(n)$ samples from $S$ produces samples $y$ from $\bar{S}$ uniformly at random, has circuit complexity $\tilde{\Omega}\left(N\right)$. 
    \label{thm:classical_circuit_lb}
\end{theorem}

\begin{proof} 
We will simply assume that all $l$ samples are distinct, since this not being true only increases the lower bound that follows from our argument. Take a subset $\bar{S}$ with Kolmogorov complexity $\mathcal{K}(\bar{S}) = N/2$. Such a set can be constructed by taking a Kolmogorov random string $z \in \{0,1\}^{N/2}$ of length $N/2$, with each element $z_i$ indexed by a string $i \in \{0,1\}^{n-1}$. We define  $y_i = (z_i,i)$ as the concatenation of bit $z_i$ and the string $i$. For each $i \in \{0,1\}^{n-1}$, we add the string $y_i$ to the set $\bar{S}$. This way, one can fully reconstruct $z$ if one knows all the strings in $\bar{S}$. Then, any sampler $M$ that produces samples from $\bar{S}$ uniformly at random, given the description of any $l$ elements from $S$, is a description of the set $\bar{S}$ (and thus of the string $z$), since one can run the sampler until all $N/2$ distinct labels from $\bar{S}$ are observed. Since $l$ elements from $S$ can be described with $l \log N$ bits, the description length of the program for $M$ given $l$ such elements has to be at least $ N/2-l \log N-\mathcal{O}(1)$. Since any circuit description (of the elementary gates) of size $s$ can represent a program for $M$ in $\mathcal{O}(s \log s)$ bits, we must have that the circuit complexity of any such circuit which implements $M$ is lower bounded by 
\begin{align*}
    \Omega\left(\frac{N/2-l \log N-\mathcal{O}(1)}{\log \left(N/2-l \log N-\mathcal{O}(1) \right)}\right) = \tilde{\Omega}\left(N\right),
\end{align*}
when $l=\poly(n)$.
\end{proof}

\section{Circuit complexity and distinguishability of subset states}
\label{sec:quantum_cc_and_d}

In the final section we will show that even though we have proven in the above sections that a quantum algorithm can perfectly swap between any $\ket{S}$ and $\ket{\bar{S}}$, there exist pairs of $\ket{S}$, $\ket{\bar{S}}$ so that it is computationally intractable to find out which one was given as an input. As an immediate corollary, this result implies that generally states of the form of $\ket{S}$ must have exponential circuit complexity. This means that in order to turn complement sampling into a suitable quantum advantage experiment, the hardness result for strong PRPs is indeed necessary. We will first prove the easier case of the exact setting and then move on the general case. We will utilize the complexity definitions given in~\cref{ss:AAS}. Again, we fix the subset cardinality to $K=N/2$ where $N=2^n$ and $n$ is the number of qubits. 

\subsection{Warm-up: exact case}

As a warm-up, we first consider the simple example where we care about being able to distinguish $\ket{S}$ and $\ket{\bar{S}}$ perfectly. Consider again the phase states
\begin{align}
    \ket{y_{f_\mathrm{con}}} = \frac{1}{\sqrt{N}} \sum_{x \in \omega}\ket{x}, \qquad \ket{y_{f_\mathrm{bal}}} = \frac{1}{\sqrt{N}}\sum_{x \in \omega} (-1)^{f_{\textup{bal}}(x)} \ket{x},
    \label{eq:x_and_yf}
\end{align}
for some Boolean function $f_{\textup{bal}}: \omega \rightarrow \{0,1\}$. Recall that 
\begin{align*}
    &\ket{S} = \frac{1}{\sqrt{2}}(\ket{y_{f_\mathrm{con}}} + \ket{y_{f_\mathrm{bal}}}), \qquad \ket{\bar{S}} = \frac{1}{\sqrt{2}}(\ket{y_{f_\mathrm{con}}} - \ket{y_{f_\mathrm{bal}}}),
\end{align*}
where
\begin{align*}
    S = \{x : f_{\textup{bal}}(x) = 0\}, \qquad \bar{S} = \{x : f_{\textup{bal}}(x) = 1\}.
\end{align*}
The total number of balanced functions is lower bounded by 
\begin{align*}
    \binom{N}{N/2} \geq 2^{\Omega(N)}.
\end{align*}
For every choice of $f_{\textup{bal}}$, a circuit can only \emph{exactly} swap $\ket{y_{f_\mathrm{con}}}$ to a single $\ket{y_{f_\mathrm{bal}}}$. Let $\mathcal{G}$ be any gate set with $g = n^{\mathcal{O}(1)}$ possible choices of gates and qubit indices on which a single (multi-qubit) gate acts. Starting from $\ket{y_{f_\mathrm{con}}}$, the total number of states that can be constructed by circuits using $M$ gates from $\mathcal{G}$ is at most $g^{M}$. Thus, for some balanced function, the corresponding circuit must have $M = \tilde{\Omega}(N)$ gates, which is the relative complexity $\mathcal{C}(\ket{y_{f_\mathrm{con}}}, \ket{y_{f_\mathrm{bal}}})$.
Since the swap complexity is lower bounded by the relative complexity~\cite{aaronson2020hardness}, there must exist a $\ket{y_{f_\mathrm{bal}}}$ such that
\begin{align}
    \mathcal{S}(\ket{y_{f_\mathrm{con}}},\ket{y_{f_\mathrm{bal}}}) \geq \tilde{\Omega}(N).
    \label{eq:exact_case_omega_N}
\end{align}
Note that any balanced phase state is orthogonal to the constant phase state
\begin{align*}
    \bra{y_{f_\mathrm{con}}}\ket{y_{f_\mathrm{bal}}} = \frac{1}{N} \sum_{x \in \omega} (-1)^{f_{\textup{bal}}(x)}= 0,
\end{align*}
so we can use~\cref{lem:S_versus_D} with $\ket{a} := \ket{y_{f_\mathrm{con}}}$ and $\ket{b}:=\ket{y_{f_\mathrm{bal}}}$. We have that if the exact swap complexity $\mathcal{S}(\ket{y_{f_\mathrm{con}}},\ket{y_{f_\mathrm{bal}}}) = T$, then the exact distinguishability complexity $\mathcal{D}(\ket{S},\ket{\bar{S}}) = \mathcal{O}(T)$. Hence, from~\cref{eq:exact_case_omega_N} it follows that
\begin{align*}
    \mathcal{D}(\ket{S},\ket{\bar{S}}) \geq \tilde{\Omega}(N).
\end{align*}

\subsection{Approximate case}

We consider the same $\ket{y_{f_\mathrm{con}}}$ and $\ket{y_{f_\mathrm{bal}}}$ as in~\cref{eq:x_and_yf}, but will now investigate the relative circuit complexity when an error $\epsilon$ is allowed. Let us consider two balanced functions $f_\mathrm{bal},g_\mathrm{bal} : \omega \rightarrow \{0,1\}$.  From now on, we omit the ``bal''-subscript when possible, and write $f = f_\mathrm{bal}$ and $g = g_\mathrm{bal}$ to ease the presentation. We define $I_{f,g} = | \{x : f(x) = g(x) \}| $ and $N_{f,g} = | \{x : f(x) \neq g(x) \}| $.  The overlap between $\ket{y_{f}}$ and $\ket{y_{g}}$ is given by
\begin{align*}
    \abs{\bra{y_{f}}\ket{y_{g}}} = \frac{1}{N}\abs{I_{f,g}-N_{f,g}} = \frac{1}{N}\abs{2I_{f,g}-N} ,
\end{align*}
using the relation $I_{f,g} + N_{f,g} = N$. If we want $\abs{\bra{y_{f}}\ket{y_{g}}}  \leq \epsilon$, then we must have that
\begin{align*}
   \frac{1}{N}\abs{2I_{f,g}-N} \leq \epsilon.
\end{align*}
We consider two cases:
\begin{enumerate}[label=(\roman*)]
    \item $I_{f,g}\geq \frac{N}{2}$. We find
    \begin{align*}
        \frac{1}{N}({2I_{f,g}-N}) \leq \epsilon \quad \text{if and only if} \quad I_{f,g} \leq \frac{(1+\epsilon)N }{2}.
    \end{align*}
    \item $I_{f,g}\leq \frac{N}{2}$. Likewise, we have
    \begin{align*}
        \frac{1}{N}({N-2I_{f,g}}) \leq \epsilon \quad \text{if and only if} \quad I_{f,g} \geq  \frac{(1-\epsilon)N}{2}.
    \end{align*}
\end{enumerate}
Combining cases (i) and (ii), we see that a necessary and sufficient condition for $\abs{\bra{y_{f}}\ket{y_{g}}}  \leq \epsilon$ is 
\begin{align}
   \frac{(1-\epsilon)N}{2} \leq I_{f,g} \leq \frac{(1+\epsilon)N}{2} .
   \label{eq:condition_fg}
\end{align}
The question now is how large the set of functions $\{f\}$ is such that for any pair $f \neq g$ we have that~\cref{eq:condition_fg} holds. We can formulate this as a problem in extremal set theory in the following way. Let $\mathcal{F} = \{F\}$ be a family containing $N/2$-subsets from $\omega$. We say that if $x \in F$, then a corresponding function $f$ satisfies $f(x)=1$. For example, if for $n=2$ we have $f(00) = 1$, $f(01) = 0$, $f(10) = 0$ and $f(11) = 1$, then $F = \{00,11\}$. Hence, if $|F \cap G| = l$ for some $F,G \in \mathcal{F}$ with associated functions $f$, $g$, respectively, we must have that by definition
\begin{align*}
    |\{x : f(x) = g(x) = 1 \}| = l.
\end{align*}
However, we also know that for both $F$ and $G$ the remaining $N/2-l$ elements in each set must all be distinct: for those we have either $f(x)=1$ and $g(x) = 0$ or $f(x)=0$ and $g(x) = 1$. For all remaining $x$ which have $f(x) = g(x) = 0$, i.e., $x \notin F$ and $x \notin G$,  we know that there are a total of
\begin{align*}
    |\{x : f(x) = g(x) = 0 \}| = N -l - 2(N/2-l) = l.
\end{align*}
Combining both, we have 
\begin{align*}
    I_{f,g} = |\{x : f(x) = g(x) = 0 \}| + |\{x : f(x) = g(x) = 1 \}| =  2l.
\end{align*}
Therefore, the intersection between two sets $F$ and $G$ completely determines the possible overlaps between the corresponding quantum states $\ket{y_{f}}$ and $\ket{y_{g}}$: we have by the condition in~\cref{eq:condition_fg} that if for all $F,G \in \mathcal{F}$
\begin{align}
   \frac{(1-\epsilon)N}{4}  \leq |F \cap G| \leq \frac{(1+\epsilon)N}{4},
   \label{eq:crit_F_and_G}
\end{align}
then for all corresponding $f,g$ we have that $\abs{\bra{y_{f}}\ket{y_{g}}} \leq \epsilon$.

We will now use results from extremal set theory to prove our circuit lower bound. Let $\mathcal{H}$ be a family of subsets of $N$ elements, and $L$ be a set of non-negative integers. We say that $\mathcal{H}$ is $k$-uniform if $|H|=k$ for each $H \in \mathcal{H}$. We have that $\mathcal{H}$ is $L$-intersecting if $|H \cap G| \in L$ for each pair of distinct $G,H \in \mathcal{H}$. The following lemma was proven by Ray-Chaudhuri and Wilson and provides an upper bound on the size of $\mathcal{H}$ knowing the cardinality of $L$.

\begin{lemma}[Ray-Chaudhuri--Wilson~\cite{ray1975t}] 
    Let $\mathcal{H}$ be a $L$-intersecting $k$-uniform family of subsets of a set of $N$ elements, where $s = |L| \leq k$. Then
    \begin{align*}
        |\mathcal{H}| \leq \binom{N}{s} .
    \end{align*}
\label[lemma]{lem:RW_thm}
\end{lemma}

In terms of $N$ and $s$, the lemma is easily shown to be tight by considering the family of all $s$-subsets of any set of size $N$ and $L=\{0,\dots,s-1\}$.

Using~\cref{lem:RW_thm}, we can lower bound the size of our family $\mathcal{F}$. At first glance, the bound of~\cref{lem:RW_thm} provides the wrong inequality: we are interested in a \emph{lower} bound on the size of $\mathcal{F}$, whilst~\cref{lem:RW_thm} provides an \emph{upper} bound. The key idea is that we can consider a family $\mathcal{H}$ for all subsets that do \emph{not} meet the condition of~\cref{eq:crit_F_and_G}, and upper bound the size of that family to lower bound the size of the remaining family of subsets. After all, for a fixed $k$ the union of $L_1$-intersecting and $L_2$-intersecting $k$-uniform families of subsets of $N$ elements results in the family of all $k$-uniform subsets if the union of $L_1$ and $L_2$ forms the set of all possible intersection sizes. In the following statements we will also use asymptotic $\omega(\cdot)$-notation, which is not to be confused with the universe $\omega$ (which does not have an argument).

\begin{lemma} 
    For $\epsilon = \omega\left(1/N\right) >0$, let $\mathcal{F}$ be an $L$-intersecting set of $N/2$-subsets of a set of $N$ elements, where 
    \begin{align*}
        L = \{ \lceil (1-\epsilon)N/4 \rceil, \lceil (1-\epsilon)N/4 \rceil  + 1,\dots, \lfloor  (1+\epsilon)N/4 \rfloor \}. 
    \end{align*}
    Then we have that
    \begin{align*}
        |\mathcal{F}| \geq \Omega(2^{N/2}) .
    \end{align*}
\label[lemma]{lem:number_int_subsets}
\end{lemma}

\begin{proof}
Let $\mathcal{H}$ be an $L'$-intersecting $N/2$-uniform family of subsets of a set of $N$ elements with $L' = \{0,1,\dots,N/2\} \setminus L$. We have that the cardinality of $L'$, i.e., $s = |L'|$, can be upper bounded as
\begin{align*}
    s &\leq 1+N/2 - ((1+\epsilon)N/4  - (1-\epsilon)N/4-2) \\
    &= N(1-\epsilon)/2+2\\
    &= 3 + \frac{N}{2} \left(1-\epsilon\right).
\end{align*}
For $\epsilon \geq \frac{6}{N}$ we have $s \leq N/2$ and, using~\cref{lem:RW_thm}, we get
\begin{align*}
    |\mathcal{H}| \leq \binom{N}{s} \leq \binom{N}{3 + \frac{N}{2} \left(1-\epsilon\right)} .
\end{align*}
This means that
\begin{align*}
    |\mathcal{F}| = \binom{N}{N/2}- |\mathcal{H}|
    \geq \binom{N}{N/2} - \binom{N}{3 + \frac{N}{2} \left(1-\epsilon\right)} 
    = \Omega(2^{N/2}) ,
\end{align*}
if $N(1-\epsilon)/2+3 < N/2$, which holds asymptotically if $\epsilon = \omega\left(1/N\right)$.
\end{proof}

This implies that for $\epsilon = \omega\left(1/N\right)  >0$ there are a total of $ \Omega(2^{N/2})$ functions $f = f_\mathrm{bal}$ such that the corresponding  balanced phase states from the set $\Psi_\mathrm{bal}^{\epsilon} = \{\ket{y_{f_\mathrm{bal}}}\}$ all have at most $\epsilon$ pairwise overlap.

\begin{theorem} 
    For any $\epsilon = \omega\left(1/N\right)$, there exists a subset $S \in \mathcal{S}_K$, $K=N/2$, such that
    \begin{align*}
        \mathcal{D}_\epsilon(\ket{S},\ket{\bar{S}}) = \tilde{\Omega}(N).
    \end{align*}
\end{theorem}

\begin{proof}
By~\cref{lem:number_int_subsets}, for any $\epsilon = \omega\left(1/N\right)$, there are $\Omega(2^{N/2})$ possible balanced functions $f_\mathrm{bal}$ such that all states $\ket{y_{f_\mathrm{bal}}}$ have pairwise overlap $\leq \epsilon$. Since any circuit can only $\epsilon$-approximately swap $\ket{y_{f_\mathrm{con}}}$ to a single state $\ket{y_{f_\mathrm{bal}}} \in \Psi_\mathrm{bal}^{\epsilon}$, we need to consider $ \Omega(2^{N/2})$ different possible circuits. Again, let $\mathcal{G}$ be any gate set with $g = n^{\mathcal{O}(1)}$ possible choices of gates and qubit indices on which a single (multi-qubit) gate acts. Starting from $\ket{y_{f_\mathrm{con}}}$, the total number of states that can be constructed by circuits using $M$ gates from $\mathcal{G}$ is at most $g^{M}$. Thus, for some balanced function, the swapper circuit must have $M = \tilde{\Omega}(N)$ gates, which implies $ \mathcal{S}_\epsilon \geq \Omega(N)$. By~\cref{lem:S_versus_D}, the result immediately follows.
\end{proof}

From~\cref{prop:circuit_complexity_from_D}, the following corollary immediately follows.

\begin{corollary}
   For any $\epsilon = \omega\left(1/N\right)$, there exists a subset $S \in \mathcal{S}_K$, $K=N/2$, such that
    \[
        \mathcal{C}_\epsilon(\ket{S}) = \tilde{\Omega}(N).
    \]
\label[corollary]{cor:circuit_lower_bound_S}
\end{corollary}

As a final remark,~\cref{cor:circuit_lower_bound_S} can also qualitatively be obtained as a corollary of~\cref{thm:classical_circuit_lb}, as any provided samples can be hardcoded into a circuit.

\end{document}